\documentclass{llncs}                  
\usepackage{amsmath} 
\usepackage{latexsym} 
\usepackage{amssymb}
\usepackage{color}
\usepackage{graphics}
\usepackage{bm}
\usepackage{color}
\usepackage{enumitem}

\spnewtheorem{notation}[theorem]{Notation}{\bfseries}{\itshape}

\newcommand{\cPlus}{\ensuremath{\bm{+}}}         
\newcommand{\cTimes}{\ensuremath{\bullet}}       
\newcommand{\cModTimes}{\ensuremath{\circ}}    
\newcommand{\cmp}[1]{\ensuremath{\overline{#1}}} 
\newcommand{\Min}{\ensuremath{\mbox{\rm Min}}}
\newcommand{\Max}{\ensuremath{\mbox{\rm Max}}}
\newcommand{\Card}{\ensuremath{\mbox{\rm Card}}}
\newcommand{\Shove}{\ensuremath{\mbox{\rm Shove}}}
\newcommand{\Sum}{\ensuremath{\mbox{\rm Sum}}}
\newcommand{\Product}{\ensuremath{\mbox{\rm Prod}}}
\newcommand{\Fin}{\ensuremath{\mbox{\rm Fin}}}

\newcommand{\lev}{\ensuremath{\mbox{\rm lev}}}
\newcommand{\Primes}{\ensuremath{\mbox{\rm Primes}}}
\newcommand{\Evens}{\ensuremath{\mbox{\rm Evens}}}
\newcommand{\Pow}{\ensuremath{\mbox{\rm Pow}}}
\newcommand{\Res}{\ensuremath{\mbox{\rm Res}}}

\newcommand{\set}[1]{\ensuremath{\{#1\}}}

\newcommand{\N}{\ensuremath{\mathbb{N}}}         
\newcommand{\bbP}{\ensuremath{\mathbb{P}}}       
\newcommand{\cO}{\ensuremath{\mathcal{O}}}       
\newcommand{\cE}{\ensuremath{\mathcal{E}}}       

\newcommand{\cI}{\ensuremath{\mathcal{I}}}       
\newcommand{\cP}{\ensuremath{\mathcal{P}}}       
\newcommand{\cU}{\ensuremath{\mathcal{U}}}       

\newcommand{\PSPACE}{\mbox{\sc PSpace}}

\newcommand{\EXPTIME}{\mbox{\sc ExpTime}}

\newcommand{\BH}{\mbox{\sc BH}}

\begin{document}
\title{Functions Definable by Numerical Set-Expressions}
\titlerunning{Arithmetic Circuits}
\author{Ian Pratt-Hartmann\inst{1} \and Ivo D\"{u}ntsch\inst{2}}
\institute{School of Computer Science,
University of Manchester,
Manchester M13 9PL, U.K.
\email{ipratt@cs.man.ac.uk}
\and
Department of Computer Science,
Brock University,
St. Catharines, ON, L2S 3A1, Canada.
\email{duentsch@brocku.ca}}
\date{}
\maketitle
\begin{abstract}
\noindent
A {\em numerical set-expression} is a term specifying a cascade of
arithmetic and logical operations to be performed on sets of
non-negative integers. If these operations are confined to the usual
Boolean operations together with the result of lifting addition to the
level of sets, we speak of {\em additive circuits}.  If they are
confined to the usual Boolean operations together with the result of
lifting addition and multiplication to the level of sets, we speak of
{\em arithmetic circuits}. In this paper, we investigate the
definability of sets and functions by means of additive and arithmetic
circuits, occasionally augmented with additional operations.

\bigskip

\end{abstract}
\section{Introduction}
\label{sec:intro}
Let $\N$ denote the set of natural numbers $\set{0, 1,\ldots}$, and
$\bbP$ its power set. Fix a countably infinite set of variables $V =
\set{x, y, z \ldots}$ to range over elements of $\bbP$, and let $\cO$
be any collection of functions $\bbP^k \rightarrow \bbP$ (for various
$k \geq 0$). A {\em numerical set-expression over} $\cO$ (for short:
an $\cO$-{\em circuit}) is an expression formed, in the expected way,
using the variables $V$, the singleton constants $\set{n}$ for $n \in
\N$, the functions in $\cO$, and the usual Boolean operators 
$\emptyset$, $\N$, $\cup$, $\cap$ and ${}^-$ (complement in $\N$). If $\tau$
is an $\cO$-circuit featuring only the variables $x_1, \ldots,
x_k$, then $\tau(x_1, \ldots, x_k)$, with variables in the indicated
order, defines a function $\bbP^k \rightarrow \bbP$ in the obvious
sense; in particular, if $\tau$ is an $\cO$-circuit with no variables,
then $\tau(\ )$ defines a set of natural numbers.  We ask: which functions and
sets are thus definable by $\cO$-circuits, for various salient
collections $\cO$?

Two operations in particular naturally suggest themselves as
candidates for inclusion in $\cO$.  Denote by $\cPlus$ and $\cTimes$
the result of lifting addition and multiplication to the algebra of
sets, thus:
\begin{equation}
s \cPlus t = \set{m+n | m \in s \mbox{ and } n \in t}; \hspace{1cm}
s \cTimes t  = \set{m \cdot n | m \in s \mbox{ and } n \in t}.
\label{eq:plusTimes}
\end{equation}
for $s, t \in \bbP$.  We call $\set{\cPlus}$-circuits {\em additive
  circuits} and $\set{\cPlus, \cTimes}$-circuits {\em arithmetic
  circuits}. In the sequel, we shall focus on arithmetic circuits and
their extensions with a range of additional operations.

Consider, for example, the (variable-free) arithmetic circuits
\begin{equation}
\Evens = \set{2} \cTimes \N \hspace{2cm}
\Primes = \cmp{\set{1}} \cap  \cmp{\ \cmp{\set{1}} \cTimes \cmp{\set{1}}\ }.
\label{eq:evensPrimes}
\end{equation}
From the above definitions, $\Evens$ defines $\set{2n \mid n \in \N}$,
the set of even numbers, while $\Primes$ defines the set of natural numbers
equal neither to $1$ nor to the product of any two numbers themselves
not equal to 1---that is, the set of primes. It follows that the
circuit $\cmp{\Primes \cPlus \Primes} \cap \Evens \cap \cmp{\set{0}
\cup \set{2}}$ defines the set of counterexamples to Goldbach's
celebrated conjecture that every even number greater than 3 is the sum
of two primes.  The functions defined by $\cO$-circuits featuring
variables are determined similarly, with the values of the variables
being given by the arguments to the functions. There is no requirement
that these arguments themselves be definable by $\cO$-circuits.

The moniker $\cO$-{\em circuits} for numerical set-expressions over
$\cO$ alludes to the `circuitry' found in computing technology, and is
suggested by the depiction of these expressions as labelled, directed
graphs, specifying a cascade of arithmetic and logical operations to
be performed on sets of numbers.  Thus, for example, the arithmetic
circuits of~\eqref{eq:evensPrimes} may be depicted as in
Fig.~\ref{fig:evensPrimes}. Each node in these graphs evaluates to a
set of numbers, representing a stage of the computation performed by
the circuit. Nodes without predecessors in these graphs are labelled
by constants (or variables) indicating the sets of numbers to which
they evaluate. Nodes with predecessors in the graph are labelled with
functions (of the appropriate arity) to be performed on the values of
their immediate predecessors; the results of these operations are then
taken to be the values of the nodes in question.  Finally, one of the
nodes---here identified by a double circle---is designated as the {\em
  circuit output}, and represents the final value computed.  The
interpretation of the graphs of Fig.~\ref{fig:evensPrimes} and their
correspondence to the arithmetic circuits of~\eqref{eq:evensPrimes}
should be obvious. The only essential difference between circuits
their graphical representations is that the latter, but not the
former, allow operations to {\em share} arguments (as illustrated in
Fig.~\ref{fig:evensPrimes}b), thus permitting a more compact
representation of the set or function being defined.  However, from
the point of view of expressive power, the two representations are
entirely equivalent.

\vspace{-0.75cm}

\begin{figure}[htb]
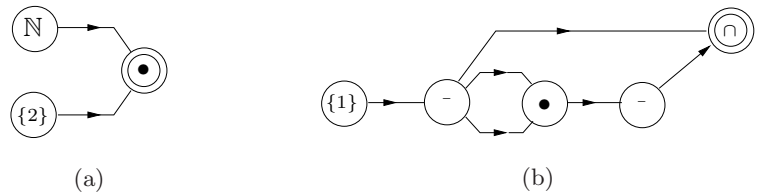

\begin{center}

\vspace{0.25cm}

\begin{minipage}{3.5cm}
\begin{center}
\input{evens.pstex_t}

\vspace{0.25cm}

(a)
\end{center}
\end{minipage}
\hspace{1cm}
\begin{minipage}{6cm}
\begin{center}
\input{primes.pstex_t}

\vspace{0.25cm}

(b)
\end{center}
\end{minipage}

\vspace{-0.5cm}

\end{center}
\label{fig:evensPrimes}
\caption{Graphical depictions of two arithmetic circuits: (a) the
  circuit $\set{2} \cTimes \N$, defining the set of even numbers; (b)
  the circuit $\cmp{\set{1}} \cap \cmp{\ \cmp{\set{1}} \cTimes
    \cmp{\set{1}}\ }$, defining the set of primes.}
\end{figure}

In principle, numerical set expressions may be considered over any
collection of functions $\cO$ with arguments and values in $\bbP$.  In
particular, any of the familiar arithmetic operations definable on
$\N$---squaring, exponentiation, (truncated) subtraction etc.---can be
lifted to the level of sets, analogously to addition and
multiplication. Other salient operations on sets of natural numbers are not
the result of lifting any arithmetic operations, however. Natural
examples are the functions
\begin{align*}
\Max(x) &=
\begin{cases}
\emptyset & \text{if $x$ empty}\\
\N            & \text{if $x$ infinite}\\
\set{\max(x)} & \text{otherwise}
\end{cases} 
&
\Card(x) &= 
\begin{cases}
\set{|x|} & \text{if $x$ finite}\\
\N & \text{otherwise}
\end{cases}\\
\ \\
\varepsilon(x) &= 
\begin{cases}
\set{0} & \text{if $x$ empty}\\
\emptyset & \text{otherwise},
\end{cases}
& \Fin(x) &=
\begin{cases}
\set{0} & \text{if $x$ finite}\\
\emptyset & \text{otherwise},
\end{cases}
\label{eq:eFin}
\end{align*}
and the function
\begin{equation*}
\Downarrow(x) = \set{m \in \N \mid \text{$\exists n \in x$ s.t.~$m \leq n$}}
\end{equation*}
(read: {\em downarrow}). The functions $\Max$ and $\Card$ return a
singleton containing, respectively, the maximum value and cardinality
of their arguments, where defined. We may think of the function
$\varepsilon(x)$ as a {\em test} for the property of emptiness, by
treating the sets $\set{0}$ and $\emptyset$ as the truth-values {\em
  true} and {\em false}, respectively; similarly for $\Fin(x)$. The
function $\Downarrow$ is a variant of $\Max$---primarily of technical
interest---which `fills in' all smaller elements. In the sequel, we
shall investigate the extra expressive power provided, in the context
of both additive and arithmetic circuits, by these functions.

Variable-free additive circuits seem first to have been studied by
Stockmeyer and Meyer~\cite{cf:s+m73}, under the name {\em integer
  expressions}.  In terms of Formal Language Theory, integer
expressions are the same as star-free regular expressions over a
1-element alphabet, where the integer $n$ stands for the string of
length $n$. Variable-free arithmetic circuits were first identified by
McKenzie and Wagner~\cite{cf:McKenzieW03,cf:McKenzieW07}). For any
collection of functions $\cO$, the {\em membership} problem for
variable-free $\cO$-circuits is as follows: given a number and a
variable-free $\cO$-circuit, determine whether that number is in the
set defined by that $\cO$-circuit.  The {\em non-emptiness} problem
for variable-free $\cO$-circuits is as follows: given an
$\cO$-circuit, determine whether the set of numbers it defines is
non-empty.  Stockmeyer and Meyer showed that, for additive circuits
(i.e. $\cO = \set{\cPlus}$), both problems are $\PSPACE$-complete. The
decidability of the membership and non-emptiness problems for
variable-free arithmetic circuits (i.e. $\cO = \set{\cPlus, \cTimes}$)
is still open. However, these problems become decidable when various
restrictions are imposed on the operators that may appear in the
circuits in question (including the Boolean operators). For a
complexity-theoretic analysis of these problems, see Meyer and
Stockmeyer, {\em op.~cit.}, McKenzie and Wagner, {\em op.~cit.},
Yang~\cite{cf:yang00} and Gla{\ss}er {\em et
  al.}~\cite{cf:HRTW07,cf:GRTW07}.

Additive circuits with variables are investigated in Je\.{z} and
Okhotin~\cite{cf:j+o08b,cf:j+o08a}. Let us call an additive circuit
{\em positive} if it does not feature any complementation operators.
Je\.{z} and Okhotin consider systems of equations
$\set{\sigma_i(\vec{x}) = \tau_i(\vec{x}) \mid 1 \leq i \leq n}$,
where the $\sigma_i$ and $\tau_i$ are positive additive circuits. They
show~\cite[Theorem~5]{cf:j+o08a} that, if $(s_1, s_2 \ldots, s_n)$ is
the unique (least, greatest) solution of some system of positive
additive circuit equations, then $s_1$ is recursive (r.e., co-r.e.);
conversely, for every recursive (r.e., co-r.e.)  set $s \in \bbP$,
there exists a system of positive additive circuit equations with a
unique (least, greatest) solution $(s, s_2 \ldots, s_n)$. In the
terminology of that paper, $s$ is {\em represented} by the system of
equations in question.  From this, Je\.{z} and Okhotin deduce that the
satisfiability problem for systems of positive additive circuit
equations is co-r.e.-complete. Similar results hold with so-called
{\em resolved} systems of equations---namely those in which $\vec{x} =
(x_1, \ldots, x_n)$ and $\sigma_i(\vec{x}) = x_i$ for all $i$ ($ 1
\leq i \leq n$).  Je\.{z} and Okhotin also show~\cite[Theorem
  3.1]{cf:j+o08b} that the family of sets representable as least
solutions of resolved systems of equations is included in $\EXPTIME$,
and moreover contains some $\EXPTIME$-hard sets.

The plan of the paper is as follows. Section~\ref{sec:preliminaries}
contains the principal definitions and technical background used in
the sequel.  Section~\ref{sec:sets} gives some examples of sets
definable by arithmetic circuits, and establishes the low
recognition-complexity of all such sets. Section~\ref{sec:defUndef}
shows that various set-functions, including $\Max$, $\Card$ and
$\Downarrow$, are not definable by arithmetic circuits, under a wide
range of extensions, and presents more restricted undefinability
results concerning the functions $\varepsilon$ and
$\Fin$. Section~\ref{sec:defNumerical} employs the results of the two
preceding sections to investigate the definability of functions $\N^k
\rightarrow \N$ by additive and arithmetic circuits.
Section~\ref{sec:additional} considers the effect of adding the
functions $\Max$, $\Card$ and $\Downarrow$ to additive and arithmetic
circuits. Section~\ref{sec:conclusion} concludes.
\section{Preliminaries}
\label{sec:preliminaries}
Recall from Section~\ref{sec:intro} that $\N$ denotes the set of
natural numbers, $\bbP$ its power set, and $V$ a countably infinite
set of variables. Henceforth, we shall refer to natural numbers simply
as {\em numbers}. We call a function (of any arity $\geq 0$) with
arguments and values in $\bbP$ a {\em set-function}.  Let $\cO$ be any
collection of set-functions. Formally, an $\cO$-{\em circuit} is
defined inductively as follows: (i) any variable in $V$ is an
$\cO$-circuit; (ii) if $\tau_1, \ldots, \tau_n$ are $\cO$-circuits ($n
\geq 0$), and $o \in \cO$ is of arity $n$, then $o(\tau_1, \ldots,
\tau_n)$ is an $\cO$-circuit; (iii) if $\tau_1$ and $\tau_2$ are
$\cO$-circuits, then so are $\tau_1 \cap \tau_2$, $\tau_1 \cup
\tau_2$, $\cmp{\tau}_1$, $\emptyset$, $\N$ and $\set{n}$, for all $n
\in \N$.  Let $o$ be any of the functions in $\cO$, or one of the
Boolean operators or singleton constants.  In keeping with the
terminology of circuitry, we speak, informally, of a {\em gate} which
{\em computes} $o$, or, more simply, of an $o$-{\em gate} to denote a
position in an $\cO$-circuit at which $o$ occurs.  We typically use
the letters $\rho$, $\sigma$, $\tau$ to range over circuits.  If
$\vec{x} = (x_1, \ldots, x_n)$ is a non-empty tuple of variables in $V$,
and $\tau$ a circuit featuring only the variables $\vec{x}$, we
optionally write $\tau$ as $\tau(\vec{x})$ to specify the order of
variables.

An {\em interpretation} is a function $\iota: V \rightarrow \bbP$
mapping variables to sets of numbers. Interpretations are extended
homomorphically to $\cO$-circuits by setting $\iota(o(\tau_1, \ldots,
\tau_k)) = o(\iota(\tau_1), \ldots, \iota(\tau_k))$ for any operator
$o$, inluding the Boolean and constant operators.  For the sake of
readability, if $\tau(\vec{x})$ is an $\cO$-circuit and $\iota$ an
interpretation mapping the tuple of variables $\vec{x}$ to the tuple
of sets of numbers $\vec{s}$, then we denote $\iota(\tau)$ by
$\tau(\vec{s})$. In particular, if $\tau$ is variable-free, the set
$\iota(\tau)$ (which is independent of $\iota$) is denoted by
$\tau(\ )$.  If $\vec{x}$ is an $n$-tuple of variables ($n>0$), the
function {\em defined by} an $\cO$-circuit $\tau(\vec{x})$ is the
function $\vec{s} \mapsto \tau(\vec{s})$.  Any function $F: \bbP^n
\rightarrow \bbP$ which can be written in this way is said to be
$\cO$-{\em definable}.  Likewise, if $\tau$ is a variable-free
$\cO$-circuit, the set {\em defined by} $\tau$ is the set of numbers
$\tau(\ )$. Any set $s \in \bbP$ which can be written in this way
is said to be $\cO$-{\em definable}.  

The following additional notation and terminology will be used.  We
write $\sigma \setminus \tau$ to abbreviate $\sigma \cap \cmp{\tau}$,
where this improves readability. Further, if, $n_1, \ldots, n_k$ are
numbers, we write $\set{n_1, \ldots, n_k}$ to denote the circuit
$\set{n_1} \cup \cdots \cup \set{n_k}$.  We omit parentheses where
possible, taking $\cTimes$ to have precendence over $\cPlus$, and
making use of the associativity of $\cap$, $\cup$, $\cPlus$ and
$\cTimes$.  If $\tau(x_, \ldots, x_n)$ is an $\cO$-circuit defining
the function $F: \bbP^n \rightarrow \bbP$, and $\sigma_1, \ldots,
\sigma_n$ are also $\cO$-circuits, we write $\tau(\sigma_1, \ldots,
\sigma_n)$ to denote the $\cO$-circuit obtained by substituting each
$\sigma_i$ for $x_i$ in $\tau$. Alternatively, where there is no
danger of confusion, we allow ourselves to write $F(\sigma_1, \ldots,
\sigma_n)$ to denote this $\cO$-circuit. (This is not strictly
correct, but obviates a lot of duplicate notation.) If $\cO_1$,
$\cO_2$ are collections of set-functions, we speak of $(\cO_1,\cO_2)$-circuits rather
than the more correct $(\cO_1 \cup \cO_2)$-circuits; likewise, if $o$
is a set-function, we speak of $(\cO_1, o)$-circuits rather than $(\cO_1 \cup
\set{o})$-circuits; and so on.

The letters $k$, $\ell$, $m$, $n$ will generally range over numbers,
and the letters $s$, $t$ over sets of numbers. Likewise, $\vec{m}$,
$\vec{n}$ will rangle over tuples of numbers, and $\vec{s}$, $\vec{t}$
over tuples of sets of numbers. We occasionally treat tuples of
numbers as sets where no confusion arises; thus, for example if
$\vec{n} = (n_1, \ldots, n_k)$, we write $\min(\vec{n})$ for
$\min(\set{n_1, \ldots, n_k})$, and so on.  For any integers $a$, $b$,
we take $[a,b]$ to denote the set $\set{a, a+1, \ldots, b}$ (empty if
$b < a$), and $[a,\infty)$, the infinite set $\set{a, a+1, \ldots }$.
  We denote the cardinality of a set of numbers $s$ by $|s|$, and
  write $s_{|m}$ for the set $s \cap [0, m]$.  If $\vec{s} = (s_1,
  \ldots, s_n)$ is a tuple of sets of numbers, we write $\vec{s}_{|m}$
  for the tuple $((s_1)_{|m}, \ldots, (s_n)_{|m})$.  If the arity is clear from
  context, $\vec{\emptyset}$ denotes the tuple $(\emptyset, \ldots,
  \emptyset)$.

We endow $\bbP$ with the topology whose basis is the collection of sets
\begin{equation*}
\set{\set{s  \cup t \mid t \subseteq [m, \infty)} \mid m \in \N, \mbox{ and }
     s \subseteq [0,m-1]},
\end{equation*}
This topology, which is compact and Hausdorff, is induced by a variety
of natural metrics, for example
\begin{equation*}
d(s, t) = 
\begin{cases}
0 & \text{if $s = t$}\\
1/(\min((s \setminus t) \cup (t \setminus s))+1) & \text{otherwise}.
\end{cases}
\end{equation*}
Hence, the metric on the product space $\bbP^n$ given by $d((s_1,
\ldots, s_n), (t_1, \ldots, t_n))$ $= \max_{1 \leq i \leq n}
d(s_i,t_i)$ induces the product topology. It is often helpful to
picture the topological space $\bbP$ in its alternative guise as {\em
  Cantor space}---the space of infinite sequences $\set{0,1}^\omega$
with basis of open sets \set{\set{\varsigma \cdot \varrho \mid \varrho \in
    \set{0,1}^\omega} \mid \varsigma \in \set{0,1}^*}---via the
bijection $\vartheta \mapsto \set{n \in \N \mid \vartheta[n] = 1}$,
where $\vartheta \in \set{0,1}^\omega$. 

The notions of continuity and uniform continuity are understood in the
usual way with respect to the above metric $d$. Specifically, $F:
\bbP^n \rightarrow \bbP$ is {\em continuous at} $\vec{s}$ if, for all
$\epsilon > 0$, there exists $\delta > 0$ such that, for all $\vec{t}
\in \bbP^n$, $d(\vec{s},\vec{t}) \leq \delta$ implies $d(F(\vec{s}),
F(\vec{t})) \leq \epsilon$; and $F$ is {\em uniformly continuous on}
$D \subseteq \bbP^n$ if, for all $\epsilon > 0$, there exists $\delta
> 0$ such that, for all $\vec{s}, \vec{t} \in D$, $d(\vec{s},\vec{t})
\leq \delta$ implies $d(F(\vec{s}), F(\vec{t})) \leq
\epsilon$. Equivalently, $F$ is continuous at $\vec{s}$ if, for all $m
\geq 0$, there exists $n \geq 0$ such that, for all $\vec{t} \in
\bbP^n$, $\vec{s}_{|n} = \vec{t}_{|n}$ implies $F(\vec{s})_{|m} =
F(\vec{t})_{|m}$; similarly for uniform continuity. The following
refinement of uniform continuity will be useful in the sequel. Suppose
$h: \N \rightarrow \N$ is a function. We say $F: \bbP^n \rightarrow
\bbP$ is $h$-{\em continuous on} $D$ if, for all $m \in \N$, and all
$\vec{s}, \vec{t} \in D$, $\vec{s}_{|h(m)} = \vec{t}_{|h(m)}$ implies
$F(\vec{s})_{|m} = F(\vec{t})_{|m}$. Thus, $F$ is uniformly continuous
on $D$ if and only if $F$ is $h$-continuous for some $h$. We are
generally interested only in the case where $h$ is {\em
  inflationary}---i.e., $h(m) \geq m$ for all $m$. If $h$ is the
identity function, $h: m \mapsto m$, we say that $F$ is {\em
  identically continuous} on $D$.  If $D$ is compact---for example, if
$D= \bbP^n$---continuity at every point of $D$ implies uniform
continuity on $D$. The converse of this statement is false; however,
if $F$ is uniformly continuous on any domain $D$, then, trivially, the
{\em restriction} of $F$ to $D$, denoted $F_{|D}$, is everywhere
continuous in $D$.

Intuitively, a continuous function is one for which the initial
segment of its value---of any desired length---can be fixed by
determining sufficiently long initial segments of its arguments.  Of
the functions encountered in Section~\ref{sec:intro}, it is routine to
check that the Boolean operations and $\cPlus$ are identically
continuous on the whole space. By contrast, $x \cTimes y$ is not
continuous at any point $(x,y) = (s, \emptyset)$ or $(x,y) = (\emptyset, s)$, where $0
\in s$, but is continuous elsewhere.  Similarly, $\varepsilon(x)$ is
discontinuous at the point $x = \emptyset$ (continuous elsewhere);
$\Downarrow(x)$ is discontinuous at $x=s$ for all finite $s$ (continuous
elsewhere); and $\Max$, $\Card$ and $\Fin$ are everywhere
discontinuous.

Some of the results obtained below concern {\em classes} of
set-functions; we introduce two important classes now. A set-function
whose values are confined to $\set{0}$ and $\emptyset$ will be
referred to as {\em predicate}; we denote the set of all predicates,
of any arity, by $\cP$. As remarked above, we are to think of
$\set{0}$ and $\emptyset$ as the truth-values {\em true} and {\em
  false}, respectively. Thus, $\varepsilon$ and $\Fin$, defined in
Section~\ref{sec:intro}, are in $\cP$; however, all the Boolean
operators and the functions $\cPlus$, $\cTimes$, $\Downarrow$, $\Max$
and $\Card$ are not. The second class of set-functions we shall be
interested in are those that are everywhere continuous---and hence, by
compactness, uniformly continuous on the whole space. We denote the
set of everywhere-continuous set-functions, of any arity, by $\cU$.
Thus, all the Boolean operators and the function $\cPlus$ are in $\cU$; however,
$\cTimes$, $\varepsilon$, $\Fin$, $\Downarrow$, $\Max$ and $\Card$ are
not.

\section{Sets definable by arithmetic circuits}
\label{sec:sets}
We begin our analysis with a brief discussion of the definability of
{\em sets} of numbers by variable-free circuits featuring the
operators introduced in Section~\ref{sec:intro}. The case of purely
additive circuits is uninteresting: a routine structural induction
shows that, if $\tau$ is a variable-free additive circuit, then
$\tau(\ )$ is finite or co-finite; conversely, every finite or
co-finite set is trivially definable by an additive circuit.  Further,
since the gates $\Downarrow$, $\Max$ and $\Card$, as well as any
predicate-gates, yield finite or co-finite outputs, these gates
obviously cannot increase the collection of definable sets.  Hence,
when discussing set-definability, we may as well restrict attention to
$(\cPlus,\cTimes)$-circuits---or, as we agreed to call them,
arithmetic circuits.

We gave two examples of such sets in Section~\ref{sec:intro}: the set
of even numbers and the set of primes, defined
in~\eqref{eq:evensPrimes} by the arithmetic circuits $\Evens$ and
$\Primes$, respectively. Other natural candidates are easy to
find. For example, if $p$ is any fixed prime $p$, the circuit
\begin{equation*}
\Pow_p = 
   \cmp{(\Primes \setminus \set{p}) \cTimes \N}
\end{equation*}
defines the set $\{p^k \mid k \in n\}$ of all powers of $p$
(see~\cite{cf:travers06}), since that is simply the set of numbers not
divisible by any prime other than $p$.  Equally evident is the fact
that, for fixed $m > k \geq 0$, the circuit
\begin{equation*}
\Res_{m,k} = \set{m} \cTimes \N + \set{k}
\end{equation*}
defines the residue class of $k$ modulo $m$. Certain other sets can be
shown to be $(\cPlus, \cTimes)$-definable, albeit less
straightforwardly. We recall the following facts of elementary number
theory (see, e.g.~Rosen~\cite[pp.~278, ff.]{mcf:rosen93}).  If $m$ and
$n$ are relatively prime integers, the congruence $m^x \equiv 1 \mod
n$ has a non-zero solution; we call the least non-zero solution $e$
the {\em order} of $m$ modulo $n$. It is a standard (and easy) result
that any other solution is divisible by $e$.
\begin{theorem} The following sets are definable by arithmetic circuits:
\begin{enumerate}[label=\textup{(}\emph{\roman*}\textup{)}]
\item the set of $k$th powers of $p$, $\{p^{nk} \mid n \in \N\}$, 
for $p$ a fixed prime and $k$ a fixed number;
\item the set of Fermat numbers, $\{2^{2^n} +1 \mid n \in \N\}$.
\end{enumerate}
\label{theo:fermat}
\end{theorem}
\begin{proof}
For the first statement, we claim that, if $k >0$ and $\ell >1$, then
$\ell^m \equiv 1 \mod \ell^k -1$ if and only if $k | m$. To see this,
observe that $\ell$ and $\ell^k -1$ are relatively prime, and that $x = k$ is
the smallest non-zero solution of the congruence $\ell^x \equiv 1 \mod
\ell^k -1$.  For $p$ a prime, the circuit 
\begin{equation*}
\Pow_p \cap \Res_{p^k-1,1}
\end{equation*}
defines the set of all numbers of the form $\set{p^m \mid p^m \equiv 1
  \mod (p^k - 1)}$. By the above claim, this is the set $\{p^{nk} \mid
n \in \N\}$.

\bigskip

For the second statement, we claim that a number of the form $2^m+1$
($m \geq 1$) is properly divisible by another number of that form if
and only if $m$ is not a power of 2. To see this, suppose first that
$m$ is not a power of 2. Write $m = a.b$ where $a \geq 3$ is the
largest odd divisor of $m$ (hence $b = 2^n$ for some $n \geq 0$). Then
\begin{equation}
2^m+1 = ((2^b)^{a-1} - (2^b)^{a-2} + \cdots +1)(2^b + 1),
\label{eq:fermat}
\end{equation}
whence $2^m+1$ is properly divisible by $2^b + 1 = 2^{2^n}+1$.
Conversely, suppose $m$ is a power of 2.  If $2^m+1$ is properly
divisible by, say, $2^\ell+1$ for some $\ell$, then
Equation~\eqref{eq:fermat} shows (substituting $\ell$ for $m$) that
$2^\ell+1$ is divisible by some Fermat number, whence $2^m+1$ is
properly divisible by some Fermat number.  But Goldbach's theorem
(see, e.g.~Rosen~\cite[p.~108]{mcf:rosen93}) states that any two
distinct Fermat numbers are in fact relatively prime. This establishes
the claim. Now, the circuit
\begin{equation*}
((\Pow_2 \setminus \set{1}) \cPlus \set{1})
   \setminus \left(((\Pow_2 \setminus 
                   \set{1})  \cPlus \set{1}) \cTimes \cmp{\set{1}}\right)
\end{equation*}
defines the set of all numbers of the form $2^m +1$ ($m \geq 1$) not
properly divisible by any other such number.
\begin{flushright}
$\qed$
\end{flushright}
\end{proof}

By numbering arithmetic circuits in some standard way
(G\"{o}del-numbering), we obtain the set $G$ of numbers $n$ such that
the circuit numbered by $n$ defines a set which does not contain
$n$. It is then routine to show that $G$ is itself not definable by
any arithmetic circuit.  However, no mathematically natural sets of
numbers have (to the authors' knowledge) been shown not to be so
definable.  Indeed, examples such as those of
Theorem~\ref{theo:fermat} give some indication of the difficulty: we
have to be sure that any candidate set cannot be defined using an
arithmetic circuit in a non-obvious way by means of some
number-theoretic fact. Nevertheless, some general facts about the
class of sets definable by arithmetic circuits can be derived: in
particular, they all have relatively low recognition-complexity.

To see why, recall our observation in Section~\ref{sec:preliminaries}
that the $\varepsilon$- and $\cTimes$-gates are discontinuous. These
facts are related. Define the function $\cModTimes: \bbP^2 \rightarrow
\bbP$ by
\begin{equation*}
s \cModTimes t = \set{ m \cdot n \mid m \in s \setminus \set{0},
                                 n \in t \setminus \set{0}}.
\end{equation*}
We see that $\cModTimes$ is identically continuous, because the
question of whether $m \in s \cModTimes t$ obviously depends only on
the initial segments $s_{|m}$ and $t_{|m}$. Furthermore:
\begin{equation}
s \cModTimes t = (s \setminus \set{0}) \cTimes (t \setminus \set{0});
\hspace{0.5cm}
s \cTimes t = 
  (s \cModTimes t) \cup 
   (\set{0} \cap ((s \setminus \varepsilon(t)) \cup (t \setminus \varepsilon(s)))).
\label{eq:cTimescModTimes}
\end{equation}
Hence, $(\cPlus, \cTimes)$-circuits and $(\cPlus, \cModTimes,
\varepsilon)$-circuits define the same sets. So therefore, do
$(\cPlus, \cTimes)$-circuits and $(\cPlus, \cModTimes)$-circuits.
By {\em bounded arithmetic}, we understand the first-order language
over the signature $(+,\cdot,1,0)$, but with all quantification
restricted to the forms $(\forall x \leq t)\varphi$ and $(\exists x
\leq t)\varphi$, where $t$ is a term.  The collection of sets in
$\N^k$ defined by formulas of bounded arithmetic is known as the {\em
  bounded hierarchy}, $\BH$ (Harrow~\cite{cf:Harrow78}).
\begin{theorem}
Every set definable by an arithmetic circuit is in $\BH$.
\label{theo:BH}
\end{theorem}
\begin{proof}
We observed above that a set is $(\cPlus, \cTimes)$-definable if and
only if it is $(\cPlus, \cModTimes)$-definable.  A routine induction
shows that every $(\cPlus, \cModTimes)$-definable set is defined by a
formula of bounded arithmetic.
\begin{flushright}
$\qed$
\end{flushright}
\end{proof}

The bounded hierarchy is known to be contained within the zeroth
Grzegorczyk class, $\cE^0_*$, and hence certainly within the class of
sets of numbers decidable in {\em deterministic} linear space---which
is equal to the second Grzegorczyk class, $\cE^2_*$
(Ritchie~\cite{cf:Ritchie63}; for a general overview, see
Rose~\cite[Ch. 5]{mcf:rose84}).  Thus, while it is not known whether
the membership problem for arithmetic circuits is decidable, the
problem of determining membership in the set $\tau(\ )$, for any fixed
arithmetic circuit $\tau$, is decidable, and indeed has relatively low
complexity.

It is interesting to relate the foregoing remarks to
language-theoretic characterizations of subsets of $\N$.  By
identifying each positive number $m$ with its binary representation as
a string in the language $1 \cdot \set{0,1}^*$, we can think of any
subset of $\N$ as a language in the usual sense of Formal Language
Theory.  (We take $0$ to be represented by the empty string.)  Under
this correspondence, we see immediately that all $(\cPlus,
\cTimes)$-definable sets are context-sensitive languages, since these
are the languages that can be recognized in {\em non-deterministic}
linear space. On the other hand, recalling
Theorem~\ref{theo:fermat}~(ii), a simple application of the pumping
lemma for context-free languages shows that the language corresponding
to the Fermat numbers---namely, $\set{10 \ldots 01 \mid \mbox{with
    $2^n -1$ zeros for some $n \geq 0$}}$---is not
context-free. Indeed, the pumping lemma of Palis and
Shende~\cite[Theorem~1]{mcf:ps95} shows that the Fermat numbers lie
outside the much larger {\em control-language hierarchy}.  We note in
passing that Theorem~\ref{theo:fermat}~(ii) is not actually necessary to
show that sets definable by arithmetic circuits are not all
context-free: the set of primes was shown not to be context-free by
Hartmanis and Shank~\cite{cf:hs68}, though this example involves a
more difficult application of the context-free pumping lemma.

\section{Definability and non-definability of set-functions}
\label{sec:defUndef}
We now turn to our principal topic: the definability of functions
$\bbP^k \rightarrow \bbP$ by arithmetic circuits and their
extensions. In this section, we pay particular attention to
limitations on definability arising from {\em finiteness} and {\em
  continuity}. 
\subsection{Functions definable by additive and arithmetic circuits}
\label{sec:defFunctions}
Many natural functions involving sets of numbers turn out to be
definable by arithmetic---or indeed additive---circuits. For example,
the function
\begin{equation*}
\downarrow(x) = \set{n \in \N \mid \forall m \in x, n \leq m}
\end{equation*}
is defined by the circuit $\tau_\downarrow(x) = \cmp{x + \N +
  \set{1}}$ (cf.~Corollary~\ref{cor:dmc}).  Likewise, the function
\begin{equation}
\Min(x) = 
\begin{cases}
\set{\min(x)} & \text{if $x \neq \emptyset$}\\
\emptyset & \text{otherwise}
\end{cases}
\label{eq:min}
\end{equation}
is defined by the circuit $\tau_\downarrow(x) \cap x$
(cf.~Corollary~\ref{cor:dmc}). 

The characteristic functions of many natural properties of sets of
numbers also turn out to be $(\cPlus, \cTimes)$-definable.  For example, if
$k$ is a number, consider the property of having cardinality greater
than $k$.  Since we have agreed to use $\set{0}$ and $\emptyset$ as
truth-values, we may take the characteristic function of this property to
be
\begin{equation*}
\Card_{> k}(x) \mapsto 
\begin{cases}
\set{0} & \text{if $|x| > k$}\\
\emptyset & \text{otherwise}.
\end{cases}
\end{equation*}
Now recursively construct the arithmetic circuits $\tau_{> k}$ as follows:
\begin{eqnarray*}
\tau_{> 0}(x) & = & x \cTimes \set{0}\\
\tau_{> k+1}(x)  & = & \tau_{> k} (x \setminus \Min(x)). 
\end{eqnarray*}
It is easy to see that $\tau_{>k}$ defines $\Card_{> k}(x)$ for all
$k$. Hence, the characteristic functions of the properties of having
cardinality at most/exactly $k$ are $(\cPlus, \cTimes)$-definable too
(cf.~Corollary~\ref{cor:finiteness}).

Definition by cases is also possible in the presence of certain collections of 
gates. We take the {\em discriminator function} to be given by
\begin{equation*}
\triangledown(x) \mapsto 
\begin{cases}
\emptyset & \text{if $x = \emptyset$}\\
\N & \text{otherwise}.
\end{cases}
\end{equation*}

\begin{lemma} 
Let $\cO$ contain $\cPlus$ and any of $\cTimes$, $\varepsilon$, $\Fin$,
  $\Max$ $\Downarrow$ or $\Card$. Then the discriminator function is
  $\cO$-definable.
\label{lma:discriminator}
\end{lemma}
\begin{proof}
The following circuits all evidently define $\triangledown$.
\begin{align*}
& x \cTimes \set{0} \cPlus \N
& & \cmp{\cmp{\Card((x \cPlus \N) \cup \set{0})} \cPlus \N} \\
& \Max(x \cPlus \N)  
& & \Downarrow(x) \cPlus \N\\
& \cmp{\varepsilon(x) \cPlus \N}
& & \cmp{\Fin(x \cPlus \N) \cPlus \N}.
\end{align*}
\begin{flushright}
$\qed$
\end{flushright}
\end{proof}
\begin{lemma}
If the functions $F, G, H: \bbP^n \rightarrow \bbP$ and
$\triangledown: \bbP \rightarrow \bbP$ are $\cO$-definable, then so is
the function
\begin{equation}
\vec{x} \mapsto
\begin{cases}
G(\vec{x}) & \text{if $F(\vec{x}) \neq \emptyset$}\\
H(\vec{x}) & \text{otherwise}.
\end{cases}
\label{eq:cases}
\end{equation}
\label{lma:cases}
\end{lemma}
\begin{proof}
Let $F$, $G$, $H$, $\triangledown$ be defined by $\rho(\vec{x})$,
$\sigma(\vec{x})$, $\tau(\vec{x})$, $\delta(x)$, respectively.  Then
the function~\eqref{eq:cases} is defined by the $\cO$-circuit:
\begin{equation*}
(\delta(\rho(\vec{x})) \cap \sigma(\vec{x})) \cup 
(\cmp{\delta(\rho(\vec{x}))} \cap \tau(\vec{x})).
\end{equation*}
\begin{flushright}
$\qed$
\end{flushright}
\end{proof}

It is also interesting to consider the $(\cPlus,
\cTimes)$-definability of functions with numbers (rather than sets of
numbers) as arguments.  We refer to such functions as {\em numerical}
functions.  We call a numerical function $f: \N^n \rightarrow \N$
$\cO$-{\em definable} if there exists an $\cO$-definable set-function $F:
\bbP^n \rightarrow \bbP$ such that, for all $m_1, \ldots, m_n$,
$F(\set{m_1}, \ldots, \set{m_n}) = \set{f(m_1, \ldots, m_n)}$.  Thus,
when discussing the definability of a numerical function, we do not
care what values any (putative) defining circuit takes on
non-singleton inputs.

Clearly, all linear functions with positive integer coefficients are
$(\cPlus)$-definable, and all polynomials with positive integer coefficients
are $(\cPlus, \cTimes)$-definable. Some other numerical functions are
definable too.  For example, the function 
\begin{equation*}
n \mapsto 
\begin{cases}
2n-1
 & \text{if $n >0$}\\
0    & \text{otherwise}
\end{cases}
\end{equation*}
is defined by the additive circuit
\begin{equation}
\Min\left(\cmp{\cmp{x \cPlus \N} \cPlus \cmp{x \cPlus \N}}\right).
\label{eq:2nMinus1}
\end{equation}
Or again, given a fixed number $\ell >1$, the function $n \mapsto (n \mod
\ell)$ is defined by the arithmetic circuit
\begin{eqnarray*}
\mod_{\ell}(x) & = & 
 \bigcup_{0 \leq k < \ell} \big(((x \cap \Res_{\ell,k}) \cTimes \set{0}) \cPlus \set{k}\big).
\end{eqnarray*}

In a similar vein, if $R \subseteq \N^n$ is an n-ary relation on
numbers, call its {\em characteristic function} the function $\chi_R$
mapping any tuple of singletons $(\set{m_1}, \ldots, \set{m_n})$ to
$\set{0}$ if $(m_1, \ldots, m_n) \in R$, and to $\emptyset$ otherwise.
Clearly, if the set $R \subseteq \N$ is defined by an arithmetic
circuit $\tau$, then $\chi_R$ is defined by the circuit $(\tau \cap x)
\cTimes \set{0}$.  Some other characteristic functions are definable
by arithmetic circuits too.  Consider, for example, the relation of
{\em relative primeness}. From the Euclidean algorithm for the
greatest common divisor, $m$ and $n$ are relatively prime if and only
if there exist integers $a$, $b$ such that $am + bn = 1$. If $m$ and
$n$ are both greater than 1, exactly one of $a$ and $b$ must be
positive and the other negative. Suppose $b$ is positive: then
$(\set{m} \cTimes \N \cPlus \set{1}) \cap (\set{n} \cTimes \N)$ is
non-empty.  Symmetrically, if $a$ is positive, then $(\set{n} \cTimes
\N \cPlus \set{1}) \cap (\set{m} \cTimes \N)$ is non-empty.  Now let
$\tau(x,y)$ be the circuit
\begin{equation}
[((x \cTimes \N \cPlus \set{1}) \cap (y \cTimes \N)) \cup 
 ((y \cTimes \N \cPlus \set{1}) \cap (x \cTimes \N))] \cTimes \set {0}.
\label{eq:coprime}
\end{equation}
It follows that, for $m >1$ and $n >1$, $\tau(\set{m},\set{n}) =
\set{0}$ if $m$, $n$ are relatively prime, and $\tau(\set{m},\set{n})
= \emptyset$ otherwise.  Taking 1 to be relatively prime to
every number, and 0 relatively prime to no number other than 1, we
observe that~\eqref{eq:coprime} yields the correct results for these
cases too.

\subsection{Definability, continuity and uniform continuity}
\label{sec:defuc}
We now proceed to establish some simple results on functions which are
not definable even by circuits with access to {\em all} predicate
gates $\cP$ and {\em all} continuous gates $\cU$.
\begin{lemma}
\label{lma:L2}
Let $h: \N \rightarrow \N$ be an inflationary function, and $\cO$ a
collection of $h$-continuous set-functions. For any $\cO$-circuit
$\sigma(\vec{x})$, there exists a $k \geq 0$ such that the function
computed by $\sigma(\vec{x})$ is $h^{(k)}$-continuous, where $h^{(k)}$
denotes the $k$-fold iteration of $h$---i.e.~$h^{(k)}(m) =
h(\cdots(h(m)) \cdots)$, and, in particular, $h^{(0)}(m) = m$.
\label{lma:uc}
\end{lemma}
\begin{proof}
Induction on the structure of $\sigma$.  If $\sigma$ is a constant
gate or variable, we may put $k = 0$. For the inductive case, suppose
$\sigma(\vec{x}) = o(\sigma_1(\vec{x}), \ldots,
\sigma_\ell(\vec{x}))$, where the gate $o$ is $h$-continuous.  Suppose
that, for each $i$ ($1 \leq i \leq \ell$), we have $k_i$ such
\linebreak that $\vec{x}_{|h^{(k_i)}(m)} = \vec{y}_{|h^{(k_i)}(m)}$
implies $\tau_i(\vec{x})_{|m} = \tau_i(\vec{y})_{|m}$ for all
$\vec{x}, \vec{y}, m$. Setting \linebreak $k= \max(\set{k_1, \ldots,
  k_\ell})+1$, we see that $\vec{x}_{|h^{(k)}(m)} =
\vec{y}_{|h^{(k)}(m)}$ implies $\tau_i(\vec{x})_{|h(m)} =
\tau_i(\vec{y})_{|h(m)}$ for all $\vec{x}, \vec{y}, m, i$, which
implies $\sigma(\vec{x})_{|m} = \sigma(\vec{y})_{|m}$ for all
$\vec{x}, \vec{y}, m$. This completes the induction.
\begin{flushright}
$\qed$
\end{flushright}
\end{proof}
It follows immediately from Lemma~\ref{lma:uc} that the discontinuous
functions $\Downarrow$, $\Max$, $\Card$, $\varepsilon$ and $\Fin$ are
not $\cU$-definable. In the first three cases, we have a slightly
stronger non-definability result. We employ the following terminology
in the sequel. If $\tau$ is an $\cO$-circuit and $\sigma =
o(\vec{\rho})$ a sub-circuit of $\tau$, where $o \in \cP$, we call
$\sigma$ a {\em predicate sub-circuit} of $\tau$. If, in addition,
$\sigma$ is not a sub-circuit of some other predicate sub-circuit of
$\tau$, we call $\sigma$ a {\em maximal predicate sub-circuit} of
$\tau$.
\begin{theorem}
The functions $\Downarrow$, $\Max$ and $\Card$ are not
$(\cU,\cP)$-definable.
\label{theo:dmc}
\end{theorem}
\begin{proof}
Let $\tau(x)$ be a $(\cU,\cP)$-circuit: we show that it does not
define any of the functions $\Downarrow(x)$, $\Max(x)$ and $\Card(x)$.
Consider all possible substitutions of constants $\set{0}$ or
$\emptyset$ for the maximal predicate sub-circuits $\pi_1(x), \ldots,
\pi_\ell(x)$ of $\tau(x)$: in each case the resulting circuit will be
$h'$-continuous for some (inflationary) $h': \N \rightarrow \N$, by
Lemma~\ref{lma:uc}. Let $h$ be the pointwise maximum of all these
$h'$; define the sequence of numbers $\set{m_i}_{i \geq 0}$ by setting
$m_0 = 1$ and $m_{i+1}= h(m_i +1) +1$, for all $i \geq 0$; and define
the sequence of sets $\set{s_i}_{i \geq 0}$ by setting $s_i = \set{m_j
  \mid 0 \leq j \leq i}$. Since the maximal predicate sub-circuits
of $\tau(x)$ can take at most $2^\ell$ possible values, let $I$ be an
infinite set of numbers such that, for all $i, j \in I$ and all $k$
($1 \leq k \leq \ell$), $\pi_k(s_i) = \pi_k(s_j)$.  It follows that
$\tau(x)$ is $h$-continuous on the domain $D = \set{s_i \mid i \in
  I}$. Pick any $i, j \in I$ with $i < j$. By construction,
$(s_i)_{|m_{i+1}-1} = (s_j)_{|m_{i+1}-1}$, i.e.~$(s_i)_{|h(m_i+1)} =
(s_j)_{|h(m_i+1)}$.  On the other hand, $(\Downarrow(s_i))_{|m_i+1}
\neq (\Downarrow(s_j))_{|m_i+1}$, since $m_i+1$ is in the latter, but
not the former.  But this is just the statement that $\Downarrow(x)$
is not $h$-continuous on $D$. Therefore, $\tau(x)$ does not compute
$\Downarrow(x)$.

\bigskip

\noindent
To show that the functions $\Max$ and $\Card$ are also not
$h$-continuous on $D$, we again pick any $i, j \in I$ with $i < j$, so
that ~$(s_i)_{|h(m_i+1)} = (s_j)_{|h(m_i+1)}$.  The result is secured
by noting that $\Card(s_i)_{|m_i+1} \neq \Card(s_j)_{|m_i+1}$, since
the former contains $|s_i| = i+1 \leq m_i$ (since $h$ is inflationary),
but the latter does not; likewise, $\Max(s_i)_{|m_i+1} \neq
\Max(s_j)_{|m_i+1}$ since the former contains $\max(s_i) =m_i$, but
the latter does not.
\begin{flushright}
$\qed$
\end{flushright}
\end{proof}
\begin{corollary}
The functions $\Downarrow$, $\Max$ and $\Card$ are not
definable by arithmetic circuits.
\label{cor:dmc}
\end{corollary}
\begin{proof}
The gates $\cPlus$ and $\cModTimes$ are continuous; and the gate
$\cTimes$ is definable by means of $\cModTimes$ and the predicate gate
$\varepsilon$.
\begin{flushright}
$\qed$
\end{flushright}
\end{proof}

Further classes of functions may be shown not to be $(\cU,
\cP)$-definable using the same technique, for example, functions with,
as we might put it, moderately fast growth.
\begin{theorem}
Let $F: \bbP \rightarrow \bbP$ be a function such that, for $s
\in \bbP$ finite, non-empty, $F(s)$ is  non-empty with
$\max(s) \leq \min({F(s)})$. Then $F$ is not $(\cU,\cP)$-circuit
definable.
\label{theo:growth}
\end{theorem}
\begin{proof}
We use the same construction as in the proof of
Theorem~\ref{theo:dmc}, except that we set $m_{i+1} = 
h(\max(m_i+1, \min(F(s_i))))+1$ for all $i \geq 0$.  
Otherwise, the proof proceeds in
exactly the same way, noting that, for $i < j$,
$(s_i)_{| m_{i+1}-1} = (s_j)_{| m_{i+1}-1}$, but
\begin{equation*}
F(s_i)_{| \max(m_i+ 1, \min(F(s_i)))} \neq
F(s_j)_{| \max(m_i+ 1, \min(F(s_i)))},
\end{equation*}
since the set on the left-hand side contains the number
$\min(F(s_i))$, whereas the set on the right-hand side certainly
contains no number less than $\min(F(s_j)) > \min(F(s_i))$.  This
contradicts the $h$-continuity of $F$ on $D$.
\begin{flushright}
$\qed$
\end{flushright}
\end{proof}
We define the functions $\Sum, \Product: \bbP \rightarrow \bbP$ as
follows:
\begin{equation*}
\Sum(x) =
\begin{cases}
\set{\Sigma(x)} & \text{if $x$ is finite}\\
\N & \text{otherwise}
\end{cases}
\hspace{1cm}
\Product(x) = 
\begin{cases}
\set{\Pi x} & \text{if $x$ finite}\\
\N        & \text{otherwise}
\end{cases}
\end{equation*}
\begin{corollary}
The functions $\Sum$ and $\Product$ are not $(\cU, \cP)$-definable.
\end{corollary}
\begin{proof}
The function $\Sum$
satisfies the conditions of Theorem~\ref{theo:growth}. Further,
if the function $\Product$ is $(\cU, \cP)$-definable, then so is
the function $x \mapsto \Product((x \cup \set{1}) \setminus \set{0})$. But this
latter function satisfies the conditions of Theorem~\ref{theo:growth}.
\begin{flushright}
$\qed$
\end{flushright}
\end{proof}

So far, we have presented non-definability results for the functions
$\Downarrow$, $\Max$, $\Card$, $\Sum$ and $\Product$, all of which are
highly discontinuous. But what about functions which have few points
of discontinuity? One such function is
\begin{equation*}
\Shove(x) = 
\begin{cases}
\set{n - \min(x) \mid n \in x} & \text{if $x$ non-empty}\\
\emptyset & \text{otherwise},
\end{cases}
\end{equation*}
which moves all the elements of its (non-empty) argument downwards `in
parallel' so that the smallest element is 0. A routine check shows
that $\Shove(x)$ is continuous everywhere in $\bbP \setminus
\set{\emptyset}$, though not uniformly continuous on $\bbP \setminus
\set{\emptyset}$.  We use by-now familiar techniques to show that
$\Shove(x)$ is not $(\cU,\cP)$-definable; however, the construction
this time is slightly more involved.
\begin{theorem}
The function $\Shove(x)$ is not $(\cU,\cP)$-definable. 
\label{theo:shove}
\end{theorem}
\begin{proof}
Let $\tau(x)$ be a $(\cU,\cP)$-circuit: we show that it does not
define $\Shove(x)$. Let the maximal predicate sub-circuits of
$\tau(x)$ be $\pi_1(x), \ldots, \pi_\ell(x)$. For $k \geq 0$, 
define $D_k$ to be the set of subsets of the
interval $[k(\ell + 2), (k+1)(\ell + 2) -1]$ that contain the smallest
element, $k(\ell + 2)$:
\begin{equation*}
D_k = 
 \set{\set{k(\ell +2)} \cup s \mid s \subseteq [k(\ell + 2)+1, (k+1)(\ell + 2) -1]}.
\end{equation*}
Let the $2^{\ell+1}$ elements of $D_k$ be listed lexicographically as
$s_{k,1}, \ldots, s_{k,2^{\ell+1}}$. Observe that, for all $k$ and $i$
($1 \leq i \leq 2^{\ell + 1}$), $\Shove(s_{k,i}) = s_{0,i}$.

\bigskip

\noindent
For $k \geq 1$, let $B_k$ be the $2^{\ell+1} \times \ell$ array of
values:
\begin{equation*}
\left(
\begin{array}{lll}
\pi_1(s_{k,1}) & \cdots & \pi_\ell(s_{k,1})\\
\vdots        &        & \vdots\\
\pi_1(s_{k,2^{\ell +1}}) & \cdots & \pi_\ell(s_{k,2^{\ell +1}})
\end{array}
\right).
\end{equation*}
Since $B_k$ can take only finitely many values, let $K$ be an infinite
set of numbers such that $B_k$ is constant as $k$ varies over
$K$. Further, since the $2^{\ell+1}$ rows of $B_i$ can take only
$2^\ell$ possible values, there certainly exist $a, b$ ($1 \leq a < b
\leq 2^{\ell+1}$) such that the rows of $B_k$ (for $k \in K$) indexed
by $a$ and $b$ are identical. Let $D = \set{s_{k,i} \mid k \in K, i
  \in \set{a,b}}$. Thus, for all $i$ ($1 \leq i \leq \ell$), the
predicate sub-circuit value $\pi_i(s)$ is constant as $s$ ranges over
the domain $D$. By Lemma~\ref{lma:uc}, $\tau$ is uniformly continuous
on $D$.

\bigskip

\noindent
We now proceed to show that $\Shove(x)$ is not uniformly continuous
on $D$, completing the proof. For all $k \geq 1$, we have, on the one hand,
\begin{equation*}
(s_{k,a})_{|k(\ell+2) -1}  = \emptyset = (s_{k,b})_{|k(\ell+2) -1},
\end{equation*}
and, on the other, 
\begin{equation*}
\Shove(s_{k,a})_{|\ell+1} = (s_{0,a})_{|\ell+1} = s_{0,a}
\neq s_{0,b} = (s_{0,b})_{|\ell+1} = \Shove(s_{k,b})_{|\ell+1}.
\end{equation*}
Thus, there exists $m$ (namely, $m = \ell+1$) such that, for all $n$,
there exist $s, t \in D$ (namely, $s = s_{k,a}$ and $t = s_{k,b}$
for some $k \in K$ with $k \geq (n+1)/(\ell +2)$) such that $s_{|n} =
t_{|n}$ and $\Shove(s)_{|m} \neq \Shove(t)_{|m}$. But this is exactly the
  statement that $\Shove(x)$ is not uniformly continuous on $D$.
\begin{flushright}
$\qed$
\end{flushright}
\end{proof}

\subsection{Undefinability results for predicates}
\label{sec:defPredicates}
The results of Section~\ref{sec:defuc} apply to circuits featuring any
predicate gates whatsoever, and thus cannot be used to show the
undefinability of one predicate in terms of others.  In this section
we turn our attention to this problem.

\begin{lemma}
Let $\vec{s}_0$ be a tuple of finite sets and $m$ a number greater
than or equal to any element of any of these sets. Let
$\sigma_1(\vec{x}), \ldots, \sigma_p(\vec{x})$ be a collection of
$\cU$-circuits. Then there exists a tuple of finite sets $\vec{s}^*$
with $\vec{s}_0 = \vec{s}^*_{|m}$, and a number $m^*$ greater than or
equal to any element of any of the sets in $\vec{s}^*$, such that, for
all $\vec{t}$ with $\vec{t}_{|m^*} = \vec{s}^*$ and all $k$
\textup{(}$1 \leq k \leq p$\textup{)}, $\sigma_k(\vec{t}) = \emptyset$
if and only if $\sigma_k(\vec{s}^*) = \emptyset$, and furthermore, if
$\sigma_k(\vec{t}) \neq \emptyset$, then $\min(\sigma_k(\vec{t})) \leq
m^*$.
\label{lma:saturate}
\end{lemma}
\begin{proof}
By Lemma~\ref{lma:uc}, let $h: \N \rightarrow \N$ be an (inflationary)
function such that the functions computed by the $\sigma_k(\vec{x})$
are all $h$-continuous. Define, for any tuple $\vec{s}$,
\begin{eqnarray*}
I(\vec{s}) & = & 
   \set{k \mid \sigma_k(\vec{s}) \neq \emptyset}\\
M(\vec{s}) & = & 
\begin{cases}
   \max(\set{h(\min(\sigma_k(\vec{\vec{s}}))) \mid k \in I(\vec{s})}) & 
\text{if $I(\vec{s})$ non-empty}\\
0 & \text{otherwise}.
\end{cases}
\end{eqnarray*}
Thus, the set $I(\vec{s})$ tells us which of the $\sigma_k(\vec{s})$
are non-empty; and each of these non-empty sets contains an
element---say, $\ell_k$---such that $h(\ell_k) \leq M(\vec{s})$. If
$\vec{t}$ satisfies $\vec{t}_{|M(\vec{s})} = \vec{s}_{|M(\vec{s})}$, then, by
$h$-continuity, for any $k \in I(\vec{s})$, $\sigma_k(\vec{s})_{|\ell_k}
= \sigma_k(\vec{t})_{|\ell_k}$, whence $\ell_k \in
\sigma_k(\vec{t})_{|\ell_k}$, and hence $k \in I(\vec{t})$. That is:
$\vec{t}_{|M(\vec{s})} = \vec{s}_{|M(\vec{s})}$ implies $I(\vec{s}) \subseteq
I(\vec{t})$. Indeed, by the same argument, $I(\vec{t}_{|M(\vec{t})}) =
I(\vec{t})$; and $\vec{t}_{|M(\vec{t})}$ is of course a tuple of
finite sets.

\bigskip

\noindent
We construct sequences $\vec{s}_0, \ldots, \vec{s}_q$ and $m_0,
\ldots, m_q$,  starting with the given $\vec{s}_0$ and $m_0 = \max(m,
M(\vec{s}_0))$. Suppose $\vec{s}_i$ and $m_i$ have been defined. If,
for all $\vec{t}$, $\vec{t}_{|m_i} = \vec{s}_i$ implies $I(\vec{s}_i)
= I(\vec{t})$, set $q = i$ and stop. Otherwise, select some $\vec{t}$
such that $\vec{t}_{|m_i} = \vec{s}_i$ and $I(\vec{s}_i) \subsetneq
I(\vec{t})$, and let $\vec{s}_{i+1} = \vec{t}_{|M(\vec{t})}$ and
$m_{i+1} = \max(m_i, M(\vec{t}))$. Since $I(\vec{s}_i)$ cannot grow
for ever, this process terminates.  It is easy to see that $\vec{s}^*
= \vec{s}_q$ and $m^* = m_q$ have the required properties.
\begin{flushright}
$\qed$
\end{flushright}
\end{proof}

\begin{theorem}
Let $F: \bbP^n \rightarrow \bbP$ be defined by a $(\cU,
\varepsilon)$-circuit.  Then there exists $\vec{s} \in \bbP^n$ and $m
\in \N$ such that $F$ is \textup{(}uniformly\textup{)} continuous on
$\set{\vec{t} \in \bbP^n \mid \vec{s}_{|m} = \vec{t}_{|m}}$.
\label{theo:eNonDefinability}
\end{theorem}
\begin{proof}
We construct a sequence $\vec{s}^{(0)}, \ldots, \vec{s}^{(d)}$ of
tuples of sets, a sequence \linebreak $m^{(0)}, \ldots, m^{(d)}$ of
numbers and a sequence $\tau^{(0)}, \ldots, \tau^{(d)}$ of
circuits. We will show that putting $\vec{s} = \vec{s}^{(d)}$ and $m =
m^{(d)}$ secures the statement of the theorem. We begin with
$\vec{s}^{(0)} = \vec{\emptyset}$, $m^{(0)} = 0$, and $\tau^{(0)} =
\tau$.

\bigskip

\noindent
Suppose $\vec{s}^{(i)}$, $m^{(i)}$ and $\tau^{(i)}$ have already been
defined. If $\tau^{(i)}$ is a $\cU$-circuit, set $d = i$, and stop the
process. Otherwise, let $\varepsilon(\sigma_1), \ldots,
\varepsilon(\sigma_p)$ be a list of the most deeply-nested
$\varepsilon$-sub-circuits of $\tau^{(i)}$. Thus, the $\sigma_k$ ($1
\leq k \leq p$) are all $\cU$-circuits.  By Lemma~\ref{lma:saturate},
we have a tuple of finite sets $\vec{s}^*$ and a number $m^*$ greater
than any element of these sets, satisfying the following properties:
(i) $\vec{s}^{(i)} = \vec{s}^*_{|m^{(i)}}$; (ii) for all $k$ ($1 \leq
k \leq p$) and all $\vec{t}$ such that $\vec{t}_{|m^*} = \vec{s}^{*}$,
$\sigma_k(\vec{t}) = \emptyset$ if and only if $\sigma_k(\vec{s}^*) =
\emptyset$. Set $\vec{s}^{(i+1)} = \vec{s}^*$, and $m^{(i+1)} =
m^*$. Further, let $\tau^{(i+1)}$ be the circuit obtained from
$\tau^{(i)}$ by substituting the constant $\set{0}$ for any
sub-circuit $\varepsilon(\sigma_k)$ such that
$\sigma_k(\vec{s}^{(i+1)}) = \emptyset$, and the constant $\emptyset$
for any sub-circuit $\varepsilon(\sigma_k)$ such that
$\sigma_k(\vec{s}^{(i+1)}) \neq \emptyset$.  We see that, for all
$\vec{t}$ such that $\vec{t}_{|m^{(i+1)}} = \vec{s}^{(i+1)}$,
$\tau^{(i+1)}(\vec{t}) = \tau^{(i)}(\vec{t})$, since the sub-circuits
$\varepsilon(\sigma_k(\vec{x}))$ of $\tau^{(i)}$ take the substituted
values ($\set{0}$ or $\emptyset$) uniformly for all such $\vec{t}$.
Since the depth of nesting of $\varepsilon$-gates in $\tau^{(i+1)}$ is
strictly less than that in $\tau^{(i)}$, this process terminates.

\bigskip

\noindent
It is simple to verify that $\vec{t}_{|m^{(d)}} = \vec{s}^{(d)}$
implies $\vec{t}_{|m^{(i)}} = \vec{s}^{(i)}$ for all $i$ ($0 \leq i
\leq d$). Hence, $\vec{t}_{|m^{(d)}} = \vec{s}^{(d)}$ implies
$\tau^{(d)}(\vec{t}) = \tau^{(d-1)}(\vec{t}) = \cdots =
\tau^{(0)}(\vec{t}) = \tau(\vec{t})$.  Since $\tau^{(d)}(\vec{x})$
computes a uniformly continuous function, $F$ is uniformly continuous
on $\set{\vec{t} \in \bbP^n \mid \vec{s}^{(d)} = \vec{t}_{|m^{(d)}}} =
\set{\vec{t} \in \bbP^n \mid (\vec{s}^{(d)})_{|m^{(d)}} =
  \vec{t}_{|m^{(d)}}}$.
\begin{flushright}
$\qed$
\end{flushright}
\end{proof}
\begin{corollary}\label{cor:fin}
The function $\Fin$ is not $(\cU,\varepsilon)$-definable.  Further, no
$(\cU,\varepsilon)$-definable function $F: \bbP \rightarrow \bbP$
satisfies any of the following conditions for all finite, non-empty
$t$: 
{\small
\begin{equation*}
F(t)  = \begin{cases} 
\set{0} & \text{\rm if $|t|$ is even}\\
\emptyset & \text{\rm otherwise};
\end{cases} \hspace{0.25cm}
F(t)  = \begin{cases} 
\set{0} & \text{\rm if $\max(t)$ even}\\
\emptyset & \text{\rm otherwise};
\end{cases}  \hspace{0.25cm}
F(t)  = \begin{cases} 
\set{0} & \text{\rm if $\sum t$ even}\\ 
\emptyset & \text{\rm otherwise}.
\end{cases} 
\end{equation*}
}
\label{cor:finiteness}
\end{corollary}
\begin{proof}
We need only verify that none of the functions in question is
uniformly continuous on any domain $D$ of the form $\set{t \in \bbP
  \mid s = t_{|m}}$, where $s$ is a finite set of numbers and $m$ a
number greater than or equal to every element of $s$.  Consider, for
example, the function $\Fin$.  Since $s$ is finite, $\Fin(s)_{|0} =
\set{0}$.  For all $n > 0$, there exists $t \in D$ (namely, $t = s
\cup [\max(m, n)+1,\infty)$) such that $s_{|n} = t_{|n}$, but
  $\Fin(t)_{|0} = \emptyset$.  This is the statement that $\Fin$ is
  not uniformly continuous on $D$.  The other functions are treated
  similarly.
\begin{flushright}
$\qed$
\end{flushright}
\end{proof}

Theorem~\ref{theo:eNonDefinability} has a different character from
Theorems~\ref{theo:dmc}--\ref{theo:shove}, since it concerns the
non-definability of one predicate in terms of another. It helps to
picture what is going on in the following terms. We remarked above
that the $\varepsilon(x)$-gate is discontinuous only at the point $x =
\emptyset$. Thus, in constructing the sequence $\vec{s}_0, \ldots,
\vec{s}_q$ in the proof of Lemma~\ref{lma:saturate}, we are
restricting attention to domains in which fewer discontinuities
remain---a process which will eventually result in a domain containing
a non-empty open set, on which the defined function is continuous. On
the other hand, the functions mentioned in
Corollary~\ref{cor:finiteness} are discontinuous at all finite,
non-empty sets, and so cannot be definable by $(\cU,
\varepsilon)$-circuits.

Arno Pauly~\cite[p.~15]{mcf:pauly09} has kindly pointed out that
Theorem~\ref{theo:eNonDefinability} is in fact a special case of a
more general theorem on discontinuous functions proved by
Hertling~\cite{mcf:hertling96}. Let $X$ and $Y$ be topological spaces
and $F:X \rightarrow Y$ a function.  If $A \subseteq X$, denote by
$F_{|A}$ the restriction of $F$ to $A$.  For any ordinal $\beta$,
define
\begin{equation*}
A_\beta =   
\begin{cases}
X & \text{if $\beta = 0$}\\ 
\set{x \in A_\alpha \mid F_{|A_\alpha} \text{not continuous at $x$}} & 
    \text{if $\beta = \alpha+1$}\\
\cap_{\alpha < \beta} A_\alpha  & \text{if $\beta$ a limit ordinal}
\end{cases}
\end{equation*}
We then say that the {\em level} of $F$, denoted $\lev^1(F)$, is the
smallest ordinal $\beta$ such that $A_\beta$ is empty, and undefined
if no such $\beta$ exists. (The superscript $1$ is used to distinguish
$\lev^1(F)$ from a related notion which we do not need here.) Thus,
for instance, if $F$ is everywhere continuous, and $X$ is non-empty,
then $\lev^1(F) = 1$.  Hertling shows (p.~19) that, for $G: X
\rightarrow Y$ and $F: Y \rightarrow Z$ functions with $Y$ a regular
space, if $G$ has finite level, then the composition, $F \circ G: X
\rightarrow Z$ satisfies $\lev^1(F \circ G) \leq \lev^1(F) \cdot
\lev^1(G)$.

Applying the apparatus of levels to the present case, we note that,
since $\varepsilon(x)$ is continuous everywhere except at $x =
\emptyset$, we have $A_1 = \set{\emptyset}$ and $A_2 = \emptyset$,
whence $\lev^1(\varepsilon(x)) = 2$. It follows that every
$(\cU,\varepsilon(x))$-definable circuit has a finite level. On the
other hand, since the function $\Fin(x)$ is everywhere discontinuous,
it has no level. Hence $\Fin(x)$ is not
$(\cU,\varepsilon(x))$-definable. By contrast, in
Theorems~\ref{theo:dmc}--\ref{theo:shove}, there is no requirement
that the predicate gates in $\cP$ have a finite level; and in
Theorem~\ref{theo:shove}, $\Shove(x)$ actually has level 2.

\section{Numerical functions}
\label{sec:defNumerical}
In Section~\ref{sec:defFunctions}, we saw various examples of
numerical functions definable by additive and arithmetic
circuits. Here, we present some corresponding non-definability results.
\subsection{Regressive functions}
\label{sec:regressive}
Call a function $f: \N^n \rightarrow \N$ {\em regressive} if the set
$\set{f(\vec{n}) \mid \vec{n} \in \N^n,\ f(\vec{n}) < \min(\vec{n})}$
is infinite. Alternatively, $f$ is regressive if, for all $k \geq 0$
there exists $\vec{n} \in \N^n$ such that $k \leq f(\vec{n}) \leq
\min(\vec{n})-1$.  Our first theorem says, in so many words, that
regressive functions are not definable by arithmetic circuits, even
when gates computing arbitrary predicates are available.
\begin{theorem}
Let $h: \N \rightarrow \N$ be an inflationary function and $\cO$ a
collection of gates computing $h$-continuous functions. Let $f: \N^k
\rightarrow \N$ be a function such that, for every $q \geq 0$, the set
$\set{f(\vec{n}) \mid \vec{n} \in \N^k,\ h^{(q)}(f(\vec{n})) <
  \min(\vec{n})}$ is infinite. Then $f$ is not $(\cO, \cP)$-definable.
\label{theo:small}
\end{theorem}
\begin{proof}
If $\vec{n} = (n_1, \ldots, n_k)$  is a tuple of numbers, denote by 
$[\vec{n}]$ the corresponding tuple of singletons $(\set{n_1}, \ldots, \set{n_k})$.
Suppose $\tau(\vec{x})$ is an $(\cO, \cP)$-circuit. Let the maximal
predicate sub-circuits of $\tau$ be $\pi_1(\vec{x}), \ldots,
\pi_\ell(\vec{x})$. If $\sigma(\vec{x})$ is obtained by substituting the constants
$\set{0}$ and $\emptyset$ for these circuits in any way, then $\sigma$ is 
$h^{(q')}$-continuous for some $q'$; let $q$ be the maximum of these $q'$.

\bigskip

\noindent
By hypothesis, there exists an infinite set $T$ of tuples $\vec{n}$
such that $f(\vec{n}) < h^{(q)}(\min(\vec{n}))$ for all $\vec{n} \in
T$, with the values $f(\vec{n})$ all distinct. Since the tuple
$\vec{v}([\vec{n}]) = (\pi_1([\vec{n}]), \ldots, \pi_\ell([\vec{n})])$
can only take a finite number of values as $\vec{n}$ ranges over $T$,
select an infinite subset $T' \subseteq T$ for which 
$\vec{v}([\vec{n}]) = (\kappa_1, \ldots, \kappa_\ell)$ is
constant. Since we may regard the sets $\kappa_i$ ($1 \leq i \leq
\ell$) as the {\em circuits} $\set{0}$ or $\emptyset$, let
$\tau'(\vec{x})$ be the result of substituting each constant
$\kappa_i$ for $\pi_i(\vec{x})$; thus, for $\vec{n} \in T'$,
$\tau'([\vec{n}]) = \tau([\vec{n}])$. By construction,
$\tau'(\vec{x})$ is $h^{(q)}$-continuous. Further, since $h^{(q)}$ is
inflationary, we may easily select an infinite subset $T'' \subseteq
T'$ such that $h^{(q)}$ is also increasing on the set $\set{f(\vec{n})
  \mid \vec{n} \in T''}$.

\bigskip

\noindent
Now pick $\vec{n}$ and $\vec{n}'$ from $T''$ with
$f(\vec{n}) < f(\vec{n}')$. By construction, we have
$h^{(q)}(f(\vec{n})) < \min(\vec{n})$, and indeed
$h^{(q)}(f(\vec{n})) \leq h^{(q)}(f(\vec{n}')) < \min(\vec{n}')$. 
Putting $m = f(\vec{n})$ and 
applying the $h^{(q)}$-continuity of $\tau'(x)$, we have
\begin{eqnarray*}
[\vec{n}]_{|h^{(q)}(m)} = \vec{\emptyset} = [\vec{n}']_{|h^{(q)}(m)} & \Rightarrow &
\tau'([\vec{n}])_{|m} = \tau'([\vec{n}'])_{|m}\\
& \Rightarrow & \tau([\vec{n}])_{|m} = \tau([\vec{n}'])_{|m}.
\end{eqnarray*}
But, also by construction, $\set{f(\vec{n})}_{|m} = \set{f(\vec{n})} \neq \emptyset
=  \set{f(\vec{n}')}_{|m}$. Therefore, $\tau(\vec{x})$ does not define $f$.
\begin{flushright}
$\qed$
\end{flushright}
\end{proof}

In the context of arithmetic circuits, if $f:\N^k \rightarrow \N$ is a
numerical function, we can treat it, by courtesy, as a
set-function---i.e., a type of {\em gate}---understanding it to mean
\begin{eqnarray*}
F(s, \ldots s_k) = 
\begin{cases} \set{f(n_1, \ldots, n_k)} & \text{if $s_i = \set{n_i}$ for all $i$
($1 \leq i \leq k$)}\\
\emptyset & \text{otherwise}.
\end{cases}
\end{eqnarray*}
Define $n \dot{-} 1$ to be $n-1$ if $n >0$, and $0$ otherwise. For $b
>1$, define $\log^*_b n$ to be $\log_b n$ if $n >0$, and $0$,
otherwise. If $r$ is a non-negative real number, denote by $\lceil r
\rceil$ the smallest natural number greater than or equal to $r$.
\begin{corollary}
Let $0 < a,b < 1$ and $c > 1$. Then:
\begin{enumerate}[label=\textup{(}\emph{\roman*}\textup{)}]
\item The function $n \mapsto n \dot{-} 1$ is not $(\cPlus, \cTimes, \cP)$-definable;
\item the function $n \mapsto \lceil a n \rceil$ is 
   not $(\cPlus, \cTimes, n \mapsto n \dot{-} 1, \cP)$-definable;
\item the function $n \mapsto \lceil n^b \rceil$ is not $(\cPlus,
  \cTimes, n \mapsto n \dot{-} 1, n \mapsto \lceil an \rceil,
  \cP)$-definable;
\item the function $n \mapsto \lceil \log^*_c n \rceil$ is not $(\cPlus,
  \cTimes, n \mapsto n \dot{-}1, n \mapsto \lceil an \rceil, n \mapsto \lceil n^b \rceil
  \cP)$-definable.
\end{enumerate}
\label{cor:hCont}
\end{corollary}
\begin{proof}
Recall that $\cTimes$ is $(\cModTimes,\varepsilon)$-definable.  The
function $f(n) = n \dot{-} 1$ and the collection $\cO = (\cPlus, \cModTimes)$
satisfy the conditions of Theorem~\ref{theo:small} with $h(m) = m$;
the function $f(n) = \lceil a n \rceil$ and the collection $\cO =
(\cPlus, \cModTimes, n \mapsto n \dot{-} 1)$ satisfy the conditions of
Theorem~\ref{theo:small} with $h(m) = m + 1$; and so on.
\begin{flushright}
$\qed$
\end{flushright}
\end{proof}
Note that Theorem~\ref{theo:small} fails if the condition that
$\set{f(\vec{n}) \mid \vec{n} \in \N^k,\ h^{(q)}(f(\vec{n})) <
  \min(\vec{n})}$ is infinite is replaced by the condition that
$\set{\vec{n} \mid \vec{n} \in \N^k,\ h^{(q)}(f(\vec{n})) <
  \min(\vec{n})}$ is infinite. For example, we have already seen that the
function $n \mapsto (n \mod \ell)$ is $(\cPlus, \cTimes)$-definable, for
all $\ell \geq 1$. Arithmetic circuits can compute remainders (for fixed, non-zero
divisors), but not quotients.

\subsection{Semi-regressive functions}
\label{sec:semiRegressive}
We have seen that regressive numerical functions cannot be defined by
arithmetic circuits. On the other hand, the functions $n \mapsto n$
and $n \mapsto 2n$ are trivially definable by additive circuits.
Indeed, the additive circuit in~\eqref{eq:2nMinus1} defines the
function $n \mapsto 2n-1$ for $n >0$. It is therefore natural to ask
whether any numerical functions definable by additive or arithmetic
circuits can have growth in between that of $n \mapsto n$ and $n
\mapsto 2n-1$.

For simplicity, we consider only the case of 1-place
functions. (Nothing really hinges on this restriction.)  Say that $f:
\N \rightarrow \N$ is {\em semi-regressive} if, for all $\ell \geq 0$
there exists $n \geq 0$ such that $n + \ell \leq f(n) \leq 2n-2$. We show
that semi-regressive functions are not definable by additive circuits,
even when gates computing $\Downarrow$ and arbitrary predicates are
available.
\begin{lemma}
Let $\sigma(x)$ be a $(\cPlus, \Downarrow, \cP)$-circuit.  There exists a
number $k(\sigma)$ such that, for all $m \in \N$, $\sigma(\set{m})$ is
uniform on the interval $[k(\sigma),m-1]$: that is to say, either
$\sigma(\set{m}) \supseteq [k(\sigma),m-1]$ or $\sigma(\set{m}) \cap
[k(\sigma),m-1] = \emptyset$.
\label{lma:plusConstant}
\end{lemma}
\begin{proof}
We define $k(\sigma)$ inductively.  If $\sigma$ is $x$, $\emptyset$ or
$\N$, it suffices to take $k(\sigma) = 0$.  If $\sigma$ is a predicate
circuit, it suffices to take $k(\sigma) = 1$.  If $\sigma$ is
$\set{p}$, it suffices to take $k(\sigma) = p+1$.  If $\sigma$ is
$\sigma_1 \cup \sigma_2$ or $\sigma_1 \cap \sigma_2$, it suffices to
take $k(\sigma) = \max(k(\sigma_1), k(\sigma_2))$; and if $\sigma$ is
$\cmp{\sigma_1}$ or $\Downarrow(\sigma_1)$, it suffices to take
$k(\sigma) = k(\sigma_1)$. Finally, suppose $\sigma$ is $\sigma_1 +
\sigma_2$. We examine the sixteen cases generated by the following
four binary choices.  (We rely on the inductive hypothesis to ensure
exhaustiveness for the second two cases.)
\begin{align*}
\sigma_i(\set{m}) \cap [0, k(\sigma_i) -1] \neq \emptyset
& \mbox{ or } 
\sigma_i(\set{m}) \cap [0, k(\sigma_i) -1] = \emptyset & & (i = 1,2);\\
\sigma_i(\set{m}) \supseteq [k(\sigma_i),m-1]
& \mbox{ or }
\sigma_i(\set{m}) \cap [k(\sigma_i),m-1] = \emptyset & & (i = 1,2).
\end{align*}
Routine checking shows that, in all cases, either
$\sigma(\set{m}) \supseteq [k(\sigma_1)+k(\sigma_2),m-1]$ or
$\sigma(\set{m}) \cap [k(\sigma_1) + k(\sigma_2),m-1] = \emptyset$.
Taking $k(\sigma) = k(\sigma_1)+k(\sigma_2)$ completes the induction.
\begin{flushright}
$\qed$
\end{flushright}
\end{proof}
\begin{theorem}
No $(\cPlus, \Downarrow, \cP)$-circuit defines any regressive or
semi-regressive function $\N \rightarrow \N$.
\label{theo:additiveSemiregressive}
\end{theorem}
\begin{proof}
Let $\sigma(x)$ be a $(\cPlus, \Downarrow, \cP)$-circuit.  It is
instant from Lemma~\ref{lma:plusConstant} that $\sigma(x)$ does not
define a regressive function. To complete the proof, we show that
there exists a number $\ell(\sigma)$ such that, for all $m \in \N$,
$\sigma(\set{m})$ is uniform on the interval $[m +
  \ell(\sigma),2m-2]$: that is to say, either $\sigma(\set{m})
\supseteq [m + \ell(\sigma),2m-2]$ or $\sigma(\set{m}) \cap [m +
  \ell(\sigma),2m-2] = \emptyset$.

\bigskip

\noindent
We define $\ell(\sigma)$ inductively, making use of the numbers
$k(\sigma)$ guaranteed by Lemma~\ref{lma:plusConstant}.  If $\sigma$
is $x$ or a predicate circuit, it suffices to take $\ell(\sigma) =
1$. If $\sigma$ is $\emptyset$ or $\N$, it suffices to take
$\ell(\sigma) = 0$. If $\sigma$ is $\set{p}$, it suffices to take
$\ell(\sigma) = p+1$.  If $\sigma$ is $\sigma_1 \cup \sigma_2$ or
$\sigma_1 \cap \sigma_2$, it suffices to take $\ell(\sigma) =
\max(\ell(\sigma_1), \ell(\sigma_2))$; and if $\sigma$ is
$\cmp{\sigma_1}$ or $\Downarrow(\sigma_1)$, it suffices to take
$\ell(\sigma) = \ell(\sigma_1)$. Finally, suppose $\sigma$ is
$\sigma_1 + \sigma_2$. We examine the two hundred and fifty-six cases
generated by the following eight binary choices.  (We rely on the
inductive hypothesis and the properties of $k(\sigma)$ guaranteed by
Lemma~\ref{lma:plusConstant} to ensure exhaustiveness.)  {\small
\begin{align*}
\sigma_i(\set{m}) \cap [0, k(\sigma_i) -1] \neq \emptyset
& \mbox{ or }  
\sigma_i(\set{m}) \cap [0, k(\sigma_i) -1] = \emptyset & & (i = 1,2);\\
\sigma_i(\set{m}) \supseteq [k(\sigma_i),m-1]
& \mbox{ or } 
\sigma_i(\set{m}) \cap [k(\sigma_i),m-1] = \emptyset & & (i = 1,2);\\
\sigma_i(\set{m}) \cap [m, m+\ell(\sigma_i) -1] \neq \emptyset
& \mbox{ or }  
\sigma_i(\set{m}) \cap [m, m+\ell(\sigma_i) -1] = \emptyset  & & (i = 1,2);\\
\sigma_i(\set{m}) \supseteq [m+\ell(\sigma_i),2m-2]
& \mbox{ or } 
\sigma_i(\set{m}) \cap [m+ \ell(\sigma_i),2m-2] = \emptyset  & & (i = 1,2).
\end{align*}
}
Consider, for example, any cases in which both $\sigma_1(\set{m})
\supseteq [k(\sigma_1),m-1]$ and $\sigma_2(\set{m}) \supseteq
[k(\sigma_2),m-1]$. Then we see that $\sigma(\set{m})
\supseteq [k(\sigma_1) + k(\sigma_2),2m-2]$, whence, certainly,
$\sigma(\set{m})
\supseteq [m+ k(\sigma_1) + k(\sigma_2),2m-2]$.
Or again, consider any cases in which $\sigma_1(\set{m}) \cap [0,
  k(\sigma_1) -1] = \sigma_1(\set{m}) \cap [k(\sigma_1),m-1] =
\sigma_1(\set{m}) \cap [m + \ell(\sigma_1),2m-2] = \emptyset$ and
$\sigma_2(\set{m}) \cap [k(\sigma_2),m-1] = \emptyset$. Then we see
that $\sigma(\set{m}) \cap [m + k(\sigma_2) + \ell(\sigma_1)
  -1,2m-2] = \emptyset$.  Routine (but laborious) checking shows that,
in all cases, we can find a constant $\ell(\sigma)$---expressed as some
function of $k(\sigma_1)$, $k(\sigma_2)$
$\ell(\sigma_1)$ and $\ell(\sigma_2)$, depending on the case we are
dealing with---such that either $\sigma(\set{m}) \supseteq [m +
  \ell(\sigma), 2m-2]$ or $\sigma(\set{m}) \cap [m+ \ell(\sigma),
  2m-2] = \emptyset$.  This completes the induction.
\begin{flushright}
$\qed$
\end{flushright}
\end{proof}

For arithmetic circuits, by contrast, this restriction does not apply.
Consider the function $f: \N \rightarrow \N$ which maps any number $n$
to the smallest prime greater than $n$. This function is
defined by the arithmetic circuit $\Min((x \cPlus
\N \cPlus \set{1}) \cap \Primes)$. On the
one hand, the decreasing density of primes means that there is no $k$
such that $f(n) < n+ k$ for all $n$; on the other hand, the
Bertrand-Chebyshev theorem states that $f(n) < 2n - 2$ for $n \geq 4$.
(Tighter bounds are known for larger values of $n$; see,
e.g.~Nagura~\cite{mcf:nagura52}.) We therefore have:
\begin{theorem}
Some arithmetic circuits define semi-regressive functions.
\label{theo:arithmeticSemiregressive}
\end{theorem}

\subsection{Rapidly growing functions}
We round off this section with a result about the definability of
rapidly growing functions.  Again, we observe a difference between
additive and arithmetic circuits.
\begin{theorem}
Every numerical function $f: \N \rightarrow \N$ defined by an additive
circuit is linearly bounded.
\label{theo:linear}
\end{theorem}
\begin{proof}
A simple induction shows that, if $\sigma(x)$ is an additive circuit,
then there exists a number $k$ such that, for all $m>0$,
$\sigma(\set{m})$ is uniform on the interval $[km, \infty)$.
\begin{flushright}
$\qed$
\end{flushright}
\end{proof}

For arithmetic circuits, by contrast, this restriction does not apply.
Let $p$ be a (fixed) prime. Again, we need to recall some number
theory---this time not so elementary.  Consider the congruence $m^x
\equiv 1 \mod p$, where $p$ is a prime, and $p$ does not divide $m$.
By Fermat's `little' theorem, $x = p-1$ is always a solution of this
congruence; and we say that $m$ is a {\em primitive root} mod $p$ if
$p-1$ is the smallest non-zero solution---that is, in the terminology
introduced above, if the {\em order} of $m$ mod $p$ is $p-1$. It is
known~\cite[Corollary~2]{mcf:h-b86} that, for all but at most two
exceptional primes $p$, there exist infinitely many primes $q$ such
that $p$ is a primitive root mod $q$.  Fix any non-exceptional prime
$p$. We know that there is a circuit $\Pow_p$ defining the set of
powers of $p$. Now, the circuit $\sigma(x) = x \cTimes(\N \setminus
\set{0}) + \set{1}$ satisfies the condition that, for any $n  > 1$,
$\sigma(\set{n})$ is the set of numbers congruent to $1$ mod $n$,
excepting $1$ itself.  Thus, if $q$ is a prime such that $p$ is a
primitive root mod $q$, the circuit $\Pow_p \cap ((\N \setminus
\set{0}) \cTimes \set{q} \cPlus \set{1})$ defines a non-empty set
whose smallest element is $p^{q-1}$. Consider the circuit
\begin{equation*}
\tau(x) = 
  [(\N \cTimes \set {p}) \cap x] \cup 
  [\Min(\Pow_p \cap ((\N \setminus \set{0}) \cTimes x \cPlus \set{1})) \cTimes \set{p}].
\end{equation*}
On input $x = \set{n}$, the two terms in square brackets each return a
singleton or the empty set, depending on whether $n$ is relatively
prime to $p$. If $p$ divides $n$, the first term in square brackets
returns the singleton $\set{n}$; otherwise, the second term in square
brackets returns the singleton $\set{p^{e+1}}$, where $e$ is the order
of $p$ mod $n$.  Hence, the circuit defines a numerical function $f:
\N \rightarrow \N$.  Furthermore, if $n$ is one of the infinitely many
primes such that $p$ is a primitive root mod $n$, then $f(n) =
p^n$. We have shown:
\begin{theorem}
There exists a function $f: \N \rightarrow \N$ such that $f$ is
definable by an arithmetic circuit, and not bounded by any polynomial.
\label{theo:fast}
\end{theorem}
It is interesting to ask whether any numerical functions definable by
arithmetic circuits are bounded below by an exponential function.

\section{Additive and arithmetic circuits with $\Downarrow$, $\Max$ and $\Card$}
\label{sec:additional}
Having demonstrated the undefinability of the functions $\Downarrow$,
$\Max$ and $\Card$ by means of arithmetic circuits, we next consider
what happens when gates computing them are added as primitives. Again,
we need to treat additive and arithmetic circuits separately.
\subsection{Simple definability results}
We begin with some easy definability results concerning the functions
$\Downarrow$, $\Max$ and $\Card$.
\begin{lemma}
Let $F_{-1}$ be the set-function given by $s \mapsto \set{\min(s)
  \dot{-} 1}$ for $s \neq \emptyset$ and $\emptyset \mapsto
\emptyset$. Then:
\begin{enumerate}[label=\textup{(}\emph{\roman*}\textup{)}]
\item $\Fin$, $\varepsilon$ and $\Downarrow$ are $(\cPlus, \Max)$-definable;
\item $F_{-1}$ is $(\cPlus, \Max)$-definable;
\item $\varepsilon$ and $\Fin$ are $(\cPlus, \Downarrow)$-definable;
\label{enum:lemmaCase3}
\item $\Max$ is $(\cPlus, \Downarrow, F_{-1})$-definable;
\item $F_{-1}$ is $(\cPlus, \Downarrow, \Card)$-definable;
\item $\Max$ is $(\cPlus, \Downarrow, \Card)$-definable.
\end{enumerate}
\label{lma:simpleDefs1+2}
\end{lemma}
\begin{proof}
The following equations are easy to verify:
\begin{enumerate}[label=(\roman*)]
\item 
$\Fin(s) = \set{0} \setminus \Max(s \cup \set{1})$,
$\varepsilon(s) = \Fin(s \cPlus \N)$,
and for $s$ finite, non-empty, 
$\Downarrow(s) = \cmp{\Max(s) \cPlus \N \cPlus \set{1}}$;
\item  for $s$ non-empty with $\min(s) > 0$,  
   $\set{\min(s) - 1} = \Max(\cmp{s \cPlus \N})$;
\item $\varepsilon(s) = \set{0} \setminus \Downarrow(s)$, and, for $s$ non-empty,
$\Fin(s) = \set{0} \cap \Downarrow((s \cPlus 
   \set{1}) \setminus \Downarrow(s))$;
\item for $s$ finite, $\Max(s) = F_{-1}((s \cPlus \set{1})
  \setminus \Downarrow(s))$;
\item for $s$ non-empty with $\min(s) > 0$, 
$F_{-1}(s) = \Card(\Downarrow(\Min(s)) \setminus \set{0,1})$;
\item for $s$ non-empty, $\Max(s) = \Card(\Downarrow(s) \setminus
  \set{0})$.
\end{enumerate}
To deal with the cases not covered by these equations, apply
Lemmas~\ref{lma:discriminator} and~\ref{lma:cases}.
\begin{flushright}
$\qed$
\end{flushright}
\end{proof}

\subsection{Circuits with $\Downarrow$}
Lemma~\ref{lma:simpleDefs1+2} does not tell us how to define $\Max$
in terms of $\Downarrow$ alone. With the help of the set of even
numbers, however, this is possible.  Recall the circuit $\Evens
= \set{2} \cTimes \N$ from~\eqref{eq:evensPrimes}.
\begin{theorem}
The gate $\Max$ is $(\cPlus, \cTimes, \Downarrow)$-definable.
\label{theo:simpleDef3}
\end{theorem}
\begin{proof}
For $s$ finite, non-empty, we have $(s \cPlus \set{1}) 
\setminus \Downarrow(s) = \set{\max(s) + 1}$.  For such values of $s$,
therefore, $((s \cPlus \set{1}) \setminus \Downarrow(s)) 
\cap \mbox{Evens}$ is empty if and only if
$\max(s)$ is odd.  But for $s$ finite, non-empty, with $\max(s)$ odd,
\begin{equation*}
\set{\max(s)} = s  \setminus \Downarrow(s \cap \mbox{Evens});
\end{equation*}
and for $s$ finite, non-empty, with $\max(s)$ even,
\begin{equation*}
\set{\max(s)} = s  \setminus \Downarrow(s \setminus \mbox{Evens}).
\end{equation*}
The result now follows by Lemmas~\ref{lma:discriminator}, \ref{lma:cases}
and~\ref{lma:simpleDefs1+2} (iii).
\begin{flushright}
$\qed$
\end{flushright}
\end{proof}

Lemma~\ref{lma:simpleDefs1+2} showed that, for additive circuits, the
gate $\Max$ is at least as expressive as $\Downarrow$, and
Theorem~\ref{theo:simpleDef3} showed that, for arithmetic circuits,
$\Max$ and $\Downarrow$ are as expressive as each other. On other
other hand, it is an easy consequence of earlier results that, for
additive circuits, $\Max$ is strictly more expressive than
$\Downarrow$.
\begin{corollary}
The function $\Max$ is not $(\cPlus, \Downarrow,\cP)$-definable.
\label{cor:plusOnlyMax}
\end{corollary}
\begin{proof}
From Theorem~\ref{theo:additiveSemiregressive} and
Lemma~\ref{lma:simpleDefs1+2} (ii), noting that $n \mapsto n
\dot{-} 1$ is regressive.
\begin{flushright}
$\qed$
\end{flushright}
\end{proof}
\begin{corollary}
The function $\Card$ is not $(\cPlus, \Downarrow,\cP)$-definable.
\label{cor:plusOnlyCard}
\end{corollary}
\begin{proof}
From Theorem~\ref{theo:additiveSemiregressive} and
and Lemma~\ref{lma:simpleDefs1+2} (v).
\begin{flushright}
$\qed$
\end{flushright}
\end{proof}
We shall strengthen Corollary~\ref{cor:plusOnlyCard} in 
Section~\ref{sec:additionals}.

\subsection{Arithmetic circuits with $\Downarrow$ and $\Card$}
\label{sec:additionals}
We next show that there are important gates which $\Downarrow$
(equivalently, $\Max$) still does not allow us to define, even when
added to {\em arithmetic} circuits. Again, we begin with a technical
lemma:
\begin{lemma}
Let $\sigma(\vec{x})$ be an $(\cI, \cP)$-circuit, where $\cI$ is the
set of identically continuous set-functions. Let $\pi_1, \ldots,
\pi_k$ be the maximal predicate sub-circuits of $\sigma$, and let
$\rho_1, \ldots, \rho_\ell$ be all the sub-circuits $\rho$ of $\sigma$
with the property that $\Downarrow(\rho)$ is also a sub-circuit of
$\sigma$. Let $m$ be a number, and let $\vec{s}$, $\vec{s}'$ be tuples
of sets of numbers, of the same arity as $\vec{x}$. If the conditions
\begin{enumerate}[label=\textup{(}\emph{\roman*}\textup{)}]
\item $\vec{s}_{|m} = \vec{s}'_{|m}$;
\item $\pi_i(\vec{s}) = \pi_i(\vec{s}')$ for all $i$ \textup{(}$1 \leq i \leq
  k$\textup{)};
\item $\rho_i(\vec{s}) \cap [m+1, \infty) = \emptyset$ if and only if
  $\rho_i(\vec{s}') \cap [m+1, \infty) = \emptyset$ for all $i$ \textup{(}$1 \leq i
  \leq \ell$\textup{)}.
\end{enumerate}
all hold, then $\sigma(\vec{s})_{|m} = \sigma(\vec{s}')_{|m}$.
\label{lma:downarrowAgree}
\end{lemma}
\begin{proof}
We prove the stronger statement that, for any sub-circuit
$\tau(\vec{x})$ of $\sigma$ which is not a proper sub-circuit of a
predicate sub-circuit of $\sigma$, we have $\tau(\vec{s})_{|m} =
\tau(\vec{s}')_{|m}$, proceeding by structural induction on $\tau$.
The case $\tau(\vec{x}) = x$ for some variable $x$ is immediate from
assumption (i). The case where $\tau(\vec{x})$ is
$\pi_i(\vec{x})$ for some $i$ ($1 \leq i \leq k$) is immediate from
assumption (ii).  Consider the case where $\tau(\vec{x})$ is
$\Downarrow(\rho_i(\vec{x}))$ for some $i$ ($1 \leq i \leq \ell$). By
inductive hypothesis, $\rho_i(\vec{s})_{|m} = \rho_i(\vec{s}')_{|m}$,
whence, by assumption (iii) and the definition of $\Downarrow$,
we see that
\mbox{$(\Downarrow(\rho_i(\vec{s})))_{|m} = 
(\Downarrow(\rho_i(\vec{s}')))_{|m}$}. The remaining cases are
immediate.
\begin{flushright}
$\qed$
\end{flushright}
\end{proof}
We now have the promised strengthening of Corollary~\ref{cor:plusOnlyCard}.
\begin{theorem}
The function $\Card$ is not $(\cPlus, \cTimes, \Downarrow,
\cP)$-definable.
\label{theo:card}
\end{theorem}
\begin{proof}
Let $\sigma(x)$ be any $(\cPlus, \cModTimes,
\Downarrow,\cP)$-circuit. Since $\cTimes$ is $(\cModTimes, \varepsilon)$-definable, 
it suffices to show that $\sigma(x)$ does not define $\Card$.

\bigskip

\noindent
Let $\pi_1, \ldots, \pi_k$ be the maximal predicate sub-circuits of
$\sigma$, and let $\rho_1, \ldots, \rho_\ell$ be all the sub-circuits
$\rho$ of $\sigma$ with the property that $\Downarrow(\rho)$ is also a
sub-circuit of $\sigma$. Let $m = 2^{k + \ell}$.  For all $s \in
\bbP$ and $1 \leq i \leq \ell$, let $\gamma_i(s)$ denote the truth-value
of the condition $\rho_i(s) \cap [m,\infty) = \emptyset$, remembering,
  of course, that truth-values are the sets $\set{0}$ (true) and
  $\emptyset$ (false).  Denote by $\vec{v}(s)$ the $k+ \ell$-tuple of
  truth-values $(\pi_1(s), \ldots, \pi_k(s), \gamma_1(s), \ldots,
  \gamma_\ell(s))$. Let $s_j = [m+1, m+j]$ for all $j$ ($0 \leq j \leq
  m$). Note that $s_0 = \emptyset$ and, for all $j$ ($0 \leq j \leq
  m$), $\Card(s_j) = j$. Clearly, we may pick $j$, $j'$ with $0 \leq j
  < j' \leq m$ such that $\vec{v}(s_j) = \vec{v}(s_j')$, since
  $\vec{v}(s_j)$ takes at most $2^{k+\ell}$ values.  All the
  conditions of Lemma~\ref{lma:downarrowAgree} are satisfied by $s=
  s_j$ and $s'= s_{j'}$. Hence $\sigma(s_j)_{|m} =
  \sigma(s_{j'})_{|m}$, so that $\sigma(x)$ does not define $\Card$,
  as required.
\begin{flushright}
$\qed$
\end{flushright}
\end{proof}
We remark that the only property of $\cPlus$ and $\cModTimes$ used in
the proof of Theorem~\ref{theo:card} is that they are identically
continuous. Thus, adding further identically continuous gates would
still not allow the definability of $\Card$.

We finish off with a partial undefinability result for $(\cPlus, \Card(x))$-circuits. 
Again, we begin with some technical lemmas.
\begin{lemma}
The following statements hold for all $s, t \in \bbP$:
\begin{enumerate}[label=\textup{(}\emph{\roman*}\textup{)}]
\item
If $s$ and $t$ are finite, $\Card(s \cPlus t) \leq \Card(s) \cdot \Card(t)$. 
\item
If either $s$ or $t$ is empty, $\Card(s \cPlus t) = 0$. 
\item
If $s$ is co-finite and $t$ non-empty, $\Card(\cmp{s \cPlus t}) \leq 
   \min(t) + \Card(\cmp{s})$. 
\end{enumerate}
\label{lma:cardMin}
\end{lemma}
\begin{proof}
Routine check.
\begin{flushright}
$\qed$
\end{flushright}
\end{proof}
Define the following functions:
\begin{eqnarray*}
\Card^*(s) & = &
\begin{cases}
\Card(s)         & \text{if $s$ is finite}\\
\Card(\cmp{s})   & \text{if $s$ is co-finite}\\
\text{undefined} & \text{otherwise.}
\end{cases}\\
{\min}^*(s) & = &
\begin{cases}
\min(s) & \text{if $s$ is non-empty}\\
-1      & \text{otherwise.}
\end{cases}
\end{eqnarray*}
\begin{lemma}
Let $s, t \in \bbP$ be finite or co-finite. Then:
\begin{enumerate}[label=\textup{(}\emph{\roman*}\textup{)}]
\item
$\Card^*(\cmp{s}) = \Card^*(s)$;
\item
$\Card^*(s \cup t) \leq \Card^*(s) + \Card^*(t)$;
\item
$\Card^*(s \cap t) \leq \Card^*(s) + \Card^*(t)$;
\item
$\Card^*(s \cPlus t) \leq \max(\Card^*(s) \cdot \Card^*(t),
\Card^*(s) \cPlus \min^*(t), \Card^*(t) + \min^*(s))$.
\end{enumerate}
\label{lma:cardStar}
\end{lemma}
\begin{proof}
The statements (i)--(iii) are immediate; (iv)
follows from Lemma~\ref{lma:cardMin}.
\begin{flushright}
$\qed$
\end{flushright}
\end{proof}
\begin{lemma}
Let $\tau(x)$ be a $(\cPlus, \cP)$-circuit, $s \in \bbP$ be finite or
co-finite, and $k >1$. Suppose that, for any sub-circuit $\sigma(x)$
of $\tau(x)$, $\min^*(\sigma(s)) \leq k$. Then \linebreak
 $\Card^*(\tau(s)) \leq
(k+ \Card^*(s))^{\lVert \tau \rVert}$, where ${\lVert \tau \rVert}$
denotes the total number of symbols in $\tau$.
\label{lma:minBound}
\end{lemma}
\begin{proof}
We show by structural induction that, for $\sigma(x)$ a sub-circuit of
$\tau$, \linebreak $\Card^*(\sigma(s)) \leq (k+ \Card^*(s))^{\lVert
  \sigma \rVert}$.  If $\sigma$ is a predicate sub-circuit or
any of $x$, $\emptyset$, $\N$ or
$\set{p}$ (necessarily: $p \leq k$), the statement is immediate. The
cases where $\sigma$ is any of $\sigma_1 \cup \sigma_2$,
$\sigma_1 \cap \sigma_2$, $\cmp{\sigma}_1$ or $\sigma_1 \cPlus \sigma_2$,
follow from the corresponding cases of Lemma~\ref{lma:cardStar}.
\begin{flushright}
$\qed$
\end{flushright}
\end{proof}
Define the function $\Max_{-1}(x)$ by
\begin{equation*}
\Max_{-1}(x) = 
\begin{cases}
\emptyset  & \text{if $x \subseteq \set{0}$}\\
\max(x) -1 & \text{if $x$ is finite, non-empty with $\max(x) > 0$}\\
\N         & \text{otherwise}.
\end{cases}
\end{equation*}
\begin{theorem}
Let $\tau(x)$ be a $(\cPlus, \Card)$-circuit in which no $\Card$-gate
appears within the scope of another. Then $\tau(x)$ does not define
the function $\Max_{-1}(x)$.
\label{theo:CardDownarrow}
\end{theorem}
\begin{proof}
Let $\sigma_1, \ldots, \sigma_p$ be the sub-circuits of $\tau$
appearing anywhere (not necessarily immediately) in the scope of a
$\Card$-gate.  Hence, each $\sigma_k$ is a $(\cPlus)$-circuit.
Applying Lemma~\ref{lma:saturate} with, say, $\vec{s}_0 = \set{0}$ and
$m = 1$, let $s^*$ be a finite set of numbers and $m^*$ a number,
greater than any element of $s^*$, such that, for all $k$ ($1 \leq k
\leq p$) and all $m \in [m^*+1, \infty)$, $\min^*(\sigma_k(s^* \cup
  \set{m})) \leq m^*$.  By Lemma~\ref{lma:minBound}, we have, for all
  such $k$ and $m$: $\Card^*(\sigma_k(s^* \cup \set{m})) \leq (|s^*|+1
  + k)^{\lVert \sigma_k \rVert}$. Note that the right-hand side of
  this inequality does not depend on $m$.  Thus, the tuple $\vec{v}_m
  = \langle \Card(\sigma_1(s^* \cup \set{m})), \ldots,
  \Card(\sigma_p(s^* \cup \set{m})) \rangle$ can take only finitely
  many values as $m$ ranges over $[m^*+1, \infty)$. Hence, we may pick
    an infinite subset $M \subseteq [m^*+1, \infty)$ such that
$\vec{v}_m$ is in fact constant as $m$ ranges over $M$. Let
      $\tau^*(x)$ be the result of replacing any sub-circuit of the
      form $\Card^*(\sigma(x))$ in $\tau$ by the singleton constant
circuit $\set{\Card^*(s^* \cup \set{m})}$ (which is independent of $m$ for $m \in M$),
      and let $D= \set{s^* \cup \set{m} \mid m \in M}$.  Thus,
      $\tau^*(x)$ is a $\cPlus$-circuit, and $\tau^*(t) = \tau(t)$
      for all $t \in D$.

\bigskip
\noindent
To complete the proof, we show that no $\cPlus$-circuit can define
$\Max_{-1}$ over $D$. For choose $m, m' \in M$, with $m < m'$. We
simply observe that $(s^* \cup \set{m})_{|m-1} = (s^* \cup
\set{m'})_{|m-1}$, but $\Max_{-1}(s^* \cup \set{m})_{|m-1} \neq
\Max_{-1}(s^* \cup \set{m})_{|m-1}$.
\begin{flushright}
$\qed$
\end{flushright}
\end{proof}
It is interesting to ask whether $\Downarrow$ or $\Max$ are in fact
definable by $(\cPlus, \cTimes, \Card, \cP)$-circuits, or even by
$(\cPlus, \Card)$-circuits.

\section{Conclusion}
\label{sec:conclusion}
In this paper, we have investigated the expressive power of numerical
set- \linebreak 
expressions over various families $\cO$ of set-functions (together
with the usual Boolean operators and singleton constants). We called
such expressions\linebreak
 $\cO$-circuits. Any variable-free $\cO$-circuit
defines a set of numbers, and any $\cO$-circuit with $n$ variables
defines a function $\bbP^n \rightarrow \bbP$, where $\bbP$ is the
power set of the numbers. Of particular interest are the operations
$\cPlus$ and $\cTimes$ which result from lifting ordinary addition and
multiplication to the level of sets. We called circuits featuring the
operator $\cPlus$ {\em additive circuits}, and those featuring both
$\cPlus$ and $\cTimes$, {\em arithmetic circuits}.

We considered the definability of functions by additive and arithmetic
circuits, with particular reference to the functions $\Downarrow$,
$\Max$ and $\Card$, as well as the predicates $\varepsilon$ (the test
for emptiness) and $\Fin$ (the test for finiteness).  We showed that
the functions of $\Downarrow$, $\Max$ and $\Card$ cannot be defined by
arithmetic circuits, even when arbitrary predicate gates are
available. We showed further that various predicates, including
$\Fin$, cannot be defined by any arithmetic circuits extended with
`less discontinuous' predicates, such as $\varepsilon$.

We also established related results on the definability of numerical
functions (functions $\N^n \rightarrow \N$) by means of additive and
arithmetic circuits.  We showed that no arithmetic circuit could
define any `regressive' function, even when arbitrary predicate gates
are available. We further showed that no additive circuit could
define any `semi-regressive' function, even when arbitrary predicate
gates and $\Downarrow$ are available; however we gave an example of a
semi-regressive function defined by an arithmetic circuit.  Finally,
we noted that all numerical functions defined by additive circuits are
linearly bounded, but gave an example of a numerical function defined
by an arithmetic circuit that is not polynomially bounded.

We considered the effect of adding gates computing the functions
$\Downarrow$, $\Max$ and $\Card$ to both additive and arithmetic
circuits.  We showed that, for both additive and arithmetic circuits,
$\Max$ is at least as expressive as $\Downarrow$. We further showed
that, for additive circuits, $\Max$ is in fact strictly more
expressive than $\Downarrow$, and that for arithmetic circuits, these gates
have the same expressive power. We showed that, even for arithmetic
circuits, these gates do not enable the function $\Card$ to be
defined. We finished with a partial result on the limited expressive
power of additive circuits extended with $\Card$-gates.
\section*{Acknowledgement}
\noindent
The authors gratefully acknowledge the support of the EPSRC (grant
ref.\linebreak EP/F069154).  The second author additionally
acknowledges the support of the NSERC, Canada. The authors would also like
to thank Mr.~Adam Trybus and Mr.~Yavor Nenov for help in preparing the
manuscript.
\bibliographystyle{plain}
\bibliography{mcf}
\end{document}